\documentclass[12pt,onecolumn]{IEEEtran}
\usepackage{amsthm}

	\usepackage{amsmath}
	\usepackage{graphicx, subfigure}
	\usepackage{amssymb}
	\usepackage{cases}
	\usepackage{cite}
	\usepackage{url}
	\usepackage{algorithm}
		\usepackage{algorithmic}
	\usepackage{relsize}
	\usepackage[nodisplayskipstretch]{setspace}
	\usepackage{longtable}
	\usepackage{booktabs}
	\usepackage{multirow}
	
	\newtheorem{theorem}{Theorem}
	\newtheorem{proposition}{Proposition}
	\newtheorem{lemma}{Lemma}
	\newtheorem{remark}{Remark}
	\newtheorem{definition}{Definition}

\theoremstyle{plain}
\begin{document}
	%
	\title{Dynamic Nested Clustering for Parallel PHY-Layer Processing in Cloud-RANs}
	\author{Congmin~Fan, Ying~Jun~(Angela)~Zhang, and Xiaojun Yuan

	\IEEEcompsocitemizethanks{
	\IEEEcompsocthanksitem
	This work was supported in part by the National Natural Science Foundation of China (Project number 61201262 and 61471241), the National Basic Research Program (973 program Program number 2013CB336700) and the RGC Direct Research Grant (Project number 2150828) established by The Chinese University of Hong Kong.
\IEEEcompsocthanksitem
Congmin~Fan and Ying~Jun~(Angela)~Zhang are with the Department of Information Engineering, The Chinese University of Hong Kong, HK, Email:\{fc012,yjzhang\}@ie.cuhk.edu.hk.
\IEEEcompsocthanksitem
Xiaojun Yuan is with the School of Information Science and Technology, ShanghaiTech University, Shanghai, China, Email:yuanxj@shanghaitech.edu.cn.}
\vspace{-14mm}
}

\maketitle
\begin{abstract}			
	Featured by centralized processing and cloud based infrastructure, Cloud Radio Access Network (C-RAN) is a promising solution to achieve an unprecedented system capacity in future wireless cellular networks. The huge capacity gain mainly comes from the centralized and coordinated signal processing at the cloud server. However, full-scale coordination in a large-scale C-RAN requires the processing of very large channel matrices, leading to high computational complexity and channel estimation overhead. To resolve this challenge, we exploit the near-sparsity of large C-RAN channel matrices, and derive a unified theoretical framework for clustering and parallel processing. Based on the framework, we propose a dynamic nested clustering (DNC) algorithm that not only greatly improves the system scalability in terms of baseband-processing and channel-estimation complexity, but also is amenable to various parallel processing strategies for different data center architectures. With the proposed algorithm, we show that the computation time for the optimal linear detector is greatly reduced from $O(N^3)$ to no higher than $O(N^{\frac{42}{23}})$, where $N$ is the number of RRHs in C-RAN. 
	\end{abstract}	
	
\begin{keywords}
Cloud-RAN; dynamic clustering; parallel processing
\end{keywords}

	\newpage
	\section{Introduction}	
 	The explosive growth in mobile data traffic threatens to outpace the infrastructure it relies on. To sustain the mobile data explosion with low bit-cost and high spectrum/energy efficiency, a revolutionary wireless cellular architecture, termed Cloud Radio Access Network (C-RAN), emerges as a promising solution \cite{mobile2011c}. In contrast to traditional base stations (BSs), the radio function units and baseband processing units (BBUs) in C-RANs are separated, and the latter are migrated to a centralized data center using an optical transport network with high bandwidth and low latency. This keeps the radio function units (also referred to as remote radio heads (RRHs)) light-weight, thereby allowing them to be deployed in a large number of small cells with low costs. Meanwhile, the centralized baseband allows RRHs to seamlessly cooperate with each other for flexible interference management, coordinated signal processing, etc. As such, C-RAN has been recognized as a ``future proof'' architecture that enables various key technologies including fibre-connected distributed antenna systems (DASs) and multi-cell coordination (CoMP) \cite{hadzialic2013cloud}. In this way, C-RAN holds great promise for significant system-capacity enhancement and cost reduction. 
	\par
	The exciting opportunities come hand in hand with new technical challenges. Theoretically speaking, the highest system capacity is achieved when all RRHs cooperatively form a large-scale virtual antenna array that jointly detects the users' signals. The full-scale coordination, however, requires the processing of a very large channel matrix consisting of channel coefficients from all mobile users to all RRHs. For example, the complexity of the optimal linear receiver grows cubically with the number of users/RRHs \cite{Tuchler2002minimum}. In other words, the normalized baseband processing complexity 
(normalized by the number of users/RRHs) grows quadratically as the size of the system becomes large. This fundamentally limits the scalability of the system. In addition, the full-scale joint RRH processing requires to estimate a large-scale channel matrix, causing significant channel estimation overhead. In \cite{LHALimit}, it is shown that the benefit of cooperation is fundamentally limited by the overhead of pilot-assisted channel estimation. 
	\par 
Distributed BS/antenna coordination has been extensively studied in CoMP and DAS systems \cite{shim2008block, zhang2010cooperative, Zhang2009networked, liu2012technical,  hong2013joint, Hiearchical2013liu, papa2008dynamic, gong2011joint, lee2014spectral, shi2013optimal
}. Most of the work has focused on throughput maximization \cite{shim2008block, zhang2010cooperative, Zhang2009networked, liu2012technical} and inter-interference management \cite{hong2013joint, Hiearchical2013liu, papa2008dynamic, gong2011joint, lee2014spectral, shi2013optimal
} by forming cooperative clusters among neighboring BSs/antennas. Few, if not none, of the work has considered the scalability of the baseband-processing and channel-estimation complexity when the system becomes extremely large. In reality, the preliminary C-RAN technology can already support around 10 km separation between the BBU pool and RRHs, covering 10-1000 RRH sites \cite{mobile2011c}. With such a large scale of coordination, the current DAS and CoMP schemes will become prohibitively expensive to implement. To solve the coordinated beamforming problem in C-RAN, a recent work by Shi \textit{et al}. \cite{scalable2014shi} proposes a low-complexity optimization algorithm. Even though the simulation results show that the proposed algorithm can significantly reduce the computation time, \cite{scalable2014shi} does not discuss how the complexity scales with the network size. Moreover, perfect knowledge of the entire channel matrix is required for beamforming, which is impractical for large-scale C-RAN.
	\par 
	In this paper, we endeavour to design a C-RAN baseband processing solution that can achieve the advantage of full-scale RRH coordination with a complexity that does not explode with the network size. In particular, our work is divided into the following steps.
	\begin{enumerate}
	\item Through rigorous analysis, we show that without causing noticeable performance loss, the signals can be detected by processing a sparsified channel matrix instead of the full channel matrix. In particular, we propose a threshold-based channel matrix sparsification method, which discards matrix entries if the corresponding link length (or large scale fading in general) is larger than a certain threshold. A closed-form expression is derived to relate the threshold to the signal-to-interference-and-noise ratio (SINR) loss due to matrix sparsification.  The result shows that for reasonably large networks, a vast majority of the channel coefficients can be ignored if we can tolerate a very small percentage of SINR loss. This not only opens up the possibility of significant complexity reduction, but also greatly reduces the channel estimation overhead. In practice, channel estimation overhead is mainly due to the estimation of small-scale fadings, because large-scale fadings vary at a much slower time scale. With the proposed channel matrix sparsification, we only need to estimate small-scale fadings corresponding to a small percentage of matrix entries that have not been discarded.
	
	
	\item 
	\begin{figure}[!h]
	\centering
	{\includegraphics[width=0.46\textwidth]{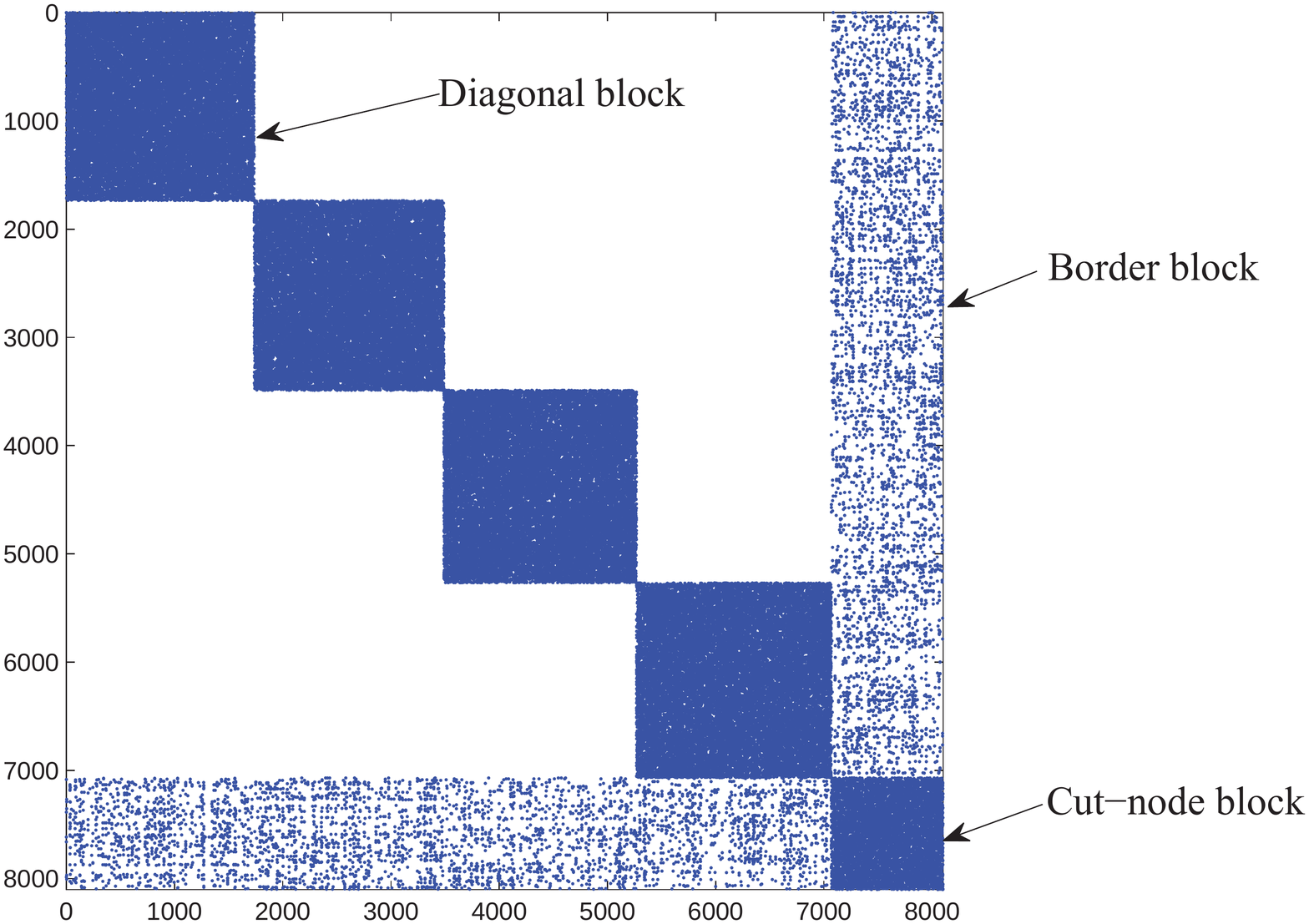}}
	\caption{A matrix in a doubly bordered block diagonal form.}\label{DBBD}
\end{figure}
	We show that by skillfully indexing the RRHs, the sparsified channel matrix can be turned into a (nested) doubly bordered block diagonal (DBBD) structure, as shown in Fig. \ref{DBBD}. Interestingly, we find that the DBBD matrix naturally leads to a dynamic nested clustering (DNC) algorithm that greatly improves the scalability of the system complexity. Specifically, the diagonal blocks (see Fig.~\ref{DBBD}) can be interpreted as clusters (or sub-networks) that are processed separately in parallel. Different clusters are coordinated by the cut-node block and border blocks that capture the interference among clusters. As such, the baseband-processing complexity is dominated by the size of the clusters instead of the entire C-RAN network.
	
	
	\item 
	Thanks to the centralized BBU pool of C-RAN, the DNC algorithm is amenable for parallel processing at the C-RAN data center. We design a parallel processing strategy that allow flexible tradeoff between the computation time, the number of parallel processors required, the allocation of computational power among BBUs, etc. In this way, the DNC algorithm is adaptive to various architectures of the BBU pool.
	\end{enumerate}
	\par 
	The rest of this paper is organized as follows. In Section II, we describe the system model and outline the steps of the DNC algorithm. The first step in the DNC algorithm, threshold-based channel sparsification, is proposed and analysed in Section III. In Section IV, a single-layer DNC algorithm is proposed, and the detailed parallel implementation of this algorithm is discussed. The multi-layer DNC algorithm is introduced in Section V. Conclusions and discussions are given in Section VI. 
%
\section{System Model}
		\vspace{-0.3cm}
	\subsection{System Setup}
	We consider the uplink transmission of a C-RAN with $N$ single-antenna RRHs and $K$ single-antenna mobile users randomly located over the entire coverage area. 
	The received signal vector $\mathbf{y}\in \mathbb{C}^{N \times 1}$ at the RRHs is
	\begin{equation}
	\small
	\mathbf{y}=\mathbf{H}\mathbf{P}^{\frac{1}{2}}\mathbf{x}+\mathbf{n},
	\end{equation}
	where $\mathbf{H} \in \mathbb{C}^{N \times K}$ denotes the channel matrix, with the $(n, k)$th entry $H_{n,k}$ being the channel coefficient between the $k$th user and the $n$th RRH; $\mathbf{P} \in \mathbb{R}^{K\times K}$ is a diagonal matrix with the $k$th diagonal entry $P_k$ being the transmitting power allocated to user $k$; $\mathbf{x} \in \mathbb{C}^{K \times 1}$ is a vector of the transmitted signal from the $K$ users and $\mathbf{n} \sim \mathcal{CN}(\mathbf{0},N_0\mathbf{I})$ is a vector of noise received by RRHs. The transmit signals are assumed to follow an independent complex Gaussian distribution with unit variance, i.e. $E[\mathbf{x}\mathbf{x}^H]=\mathbf{I}$. Further, the $(n,k)$th entry of $\mathbf{H}$ is given by $H_{n,k}= \gamma_{n,k}d_{n,k}^{-\frac{\alpha}{2}}$, where $\gamma_{n,k}$ is the i.i.d Rayleigh fading coefficient with zero mean and variance $1$, $d_{n,k}$ is the distance between the $n$th RRH and  the $k$th user, and $\alpha$  is the path loss exponent. Then, $d_{n,k}^{-\alpha}$ is the path loss from the $k$th user to the $n$th RRH. 
		\vspace{-0.4cm}
	\subsection{MMSE Detection for C-RAN}
With centralized baseband processing, a C-RAN system can jointly detect all users' signals through a full-scale RRH cooperation. Suppose that the optimal linear detection, i.e., MMSE detection, is employed. The receive beamforming matrix is 
	\begin{equation}
	\small
	\mathbf{V}=\mathbf{A}^{-1}{\mathbf H}\mathbf{P}^{\frac{1}{2}},
	\end{equation}
	where $\mathbf{A}=\mathbf{H}\mathbf{P}\mathbf{H}^H +  N_0\mathbf I$. 
	\par 
	The decision statistics of the transmitted signal vector $\mathbf{x}$ is a linear function of the received signal vector $\mathbf{y}$, i.e.
	\begin{equation}
	\small
	\widehat{\mathbf {x}}=\mathbf{V}^H\mathbf{y}=\mathbf{V}^H\mathbf{H}\mathbf{P}^{\frac{1}{2}}\mathbf{x}+\mathbf{V}^H\mathbf{n}.
	\end{equation}
	Then, the decision statistics of $x_k$ for user $k$ is 
	\begin{equation}
	\small
	\widehat{x}_k=\mathbf{v}_k^H\mathbf{h}_kP_k^{\frac{1}{2}}x_{k}+\mathbf{v}_k^H\sum_{j\neq k}\mathbf{h}_jP_j^{\frac{1}{2}}x_{j}+\mathbf{v}_k^H\mathbf{n},\label{eqn:estimatex}
	\end{equation}
	where $\mathbf{v}_k \in \mathbb{C}^{N \times 1}$ is the $k$th column of the detection matrix $\mathbf{V}$ and $\mathbf{h}_k \in \mathbb{C}^{N \times 1}$ is the $k$th column of the channel matrix $\mathbf{H}$. The SINR of user $k$ is
	\begin{equation}
	\small
	\begin{aligned}
	\text{SINR}_k=\frac{P_k|\mathbf{v}_k^H\mathbf{h}_k|^2}{\sum_{j\neq k}P_j|\mathbf{v}_k^H\mathbf{h}_j|^2 + N_0\mathbf{v}_k^H\mathbf{v}_k}.
	\end{aligned} \label{eqn:sinr}
	\end{equation}\label{sinr_mmse}
	Notice that to calculate the detection matrix $\mathbf{V}$, the full channel matrix $\mathbf{H}$ needs to be acquired and processed. That is, full channel state information (CSI) needs to be estimated at the RRHs' side. In addition, it takes as much as $O(N^3)$ operations to calculate V, because the calculation involves the inverse of an $N\times N$ matrix $\mathbf{A}$.  
As mentioned in Section I, a C-RAN generally covers a large area with a huge number of RRHs. 
As a result, the dimension of the channel matrix $\mathbf{H}$ is extremely large, and the cost of acquiring and processing $\mathbf{H}$ becomes prohibitively high. The key question is then how to obtain the best system performance by enabling a full-scale RRH cooperation without incurring high channel estimation overhead and computational complexity.  
	\vspace{-0.3cm}  
	\subsection{Sketch of the Proposed Approach}
	As described in the Introduction section, the work in this paper consists of the following steps:
	\begin{enumerate}
	\item Channel sparsification based on a link-distance threshold. 
\item Transforming the MMSE detection into a system of linear equations defined by a (nested) DBBD matrix.
\item A parallel detection algorithm based on dynamic nested clustering.
\end{enumerate}	 
	\par 
	We first introduce the channel-matrix sparsification approach in Section III. The transformation of the MMSE detection and parallel detection algorithms are discussed in both Section IV and V.
		\vspace{-0.3cm}
\section{Threshold-based Channel Sparsification}
In this section, we discuss the first step in the DNC algorithm, i.e., threshold-based channel sparsification. We first present the detailed sparsification approach in Subsection A. Then, a closed-form expression of the matrix sparsity as a function of the tolerable SINR loss is derived in Subsection B. Finally, verifications and discussions are given in Subsection C. 
\vspace{-0.3cm}
\subsection{Sparsification Approach}
	Since the RRHs and users are distributed over a large area, an RRH can only receive reasonably strong signals from a small number of nearby users, and vice versa. Thus, ignoring the small entries in $\mathbf{H}$ would significantly sparsify the matrix, hopefully with a negligible loss in system performance. In this paper, we propose to ignore the entries of $\mathbf{H}$ based on the distance between RRHs and users. In other words, the entry $H_{n,k}$ is set to $0$ when the link distance $d_{n,k}$ is larger than a threshold $d_0$\footnote{In this paper, for simplicity of analysis, we only consider path loss and Rayleigh fading but ignore shadowing. By changing the threshold from distance to the large scale fading coefficient, we can easily extend this approach to the case with shadowing.}. The resulting sparsified channel matrix, denoted by $\mathbf{\widehat{H}}$, is given by
	\begin{equation}
	\small
	{\widehat  H}_{n,k} =\begin{cases}{H}_{n,k},&d_{n,k}<d_0
	\\0, &\text{otherwise.}
	\end{cases}
	\end{equation}
	\par
	Note that we propose to sparsify the channel matrix based on the link distance instead of the actual channel coefficients that are affected by both the link distances and fast channel fadings. In practice, link distances vary much more slowly than fast channel fading. As such, the amount of overhead to estimate the link distances is negligible compared with the overhead to estimate the fast channel fadings. With the distance-based channel sparsification, instantaneous channel estimation (i.e., estimation of the fast channel fading coefficients) is needed only on short links. Otherwise, if we sparsify $\mathbf{H}$ based on the absolute values of the entries, then channel fading needs to be estimated on every link. 
	\par 
	The received signal $\mathbf{y}$ can now be represented as
	\begin{equation}
	\small
	\mathbf{y}=\mathbf{\widehat{H}}\mathbf{P}^{\frac{1}{2}}\mathbf{x}+\mathbf{\widetilde{H}}\mathbf{P}^{\frac{1}{2}}\mathbf{x}+\mathbf{n},
	\label{eqn:y}
	\end{equation}
		\vspace{-0.1cm}
	where $\mathbf{\widetilde{H}}=\mathbf{H}-\mathbf{\widehat{H}}$  consists of the entries that have been ignored. Treating the first term in the RHS of (\ref{eqn:y}) as signal, and the remaining two terms as interference plus noise, the detection matrix becomes
	\begin{equation}
	\small
	\widehat{\mathbf{V}} = \widehat{\mathbf{A}}^{-1}{\widehat{\mathbf H}}\mathbf{P}^{\frac{1}{2}},\label{eqn:v_hat}
	\end{equation}
	\vspace{-0.1cm}
	where $\widehat{\mathbf{A}}=\widehat{\mathbf{H}}\mathbf{P}\widehat{\mathbf{H}}^H + \mathbf{\Gamma} + N_0\mathbf I$, and
	\begin{equation}
	\small
	\mathbf{\Gamma}=\mathrm{E}\left[\sum_{j=1}^{K} P_j\left(\widetilde{\mathbf{h}}_j\widehat{\mathbf{h}}_j^H+\widehat{\mathbf{h}}_j\widetilde{\mathbf{h}}_j^H+\widetilde{\mathbf{h}}_j\widetilde{\mathbf{h}}_j^H\right)\right],
	\end{equation}
	with $\widehat{\mathbf{h}}_k$ and $\widetilde{\mathbf{h}}_k$ being the $k$th column of $\widehat{\mathbf{H}}$ and $\widetilde{\mathbf{H}}$ respectively.
\par 
	With this, the SINR becomes
	\begin{equation}
	\small
	\begin{aligned}
	\widehat{\text{SINR}}_k(d_0)=\frac{P_k|\widehat{\mathbf{v}}_k^H\mathbf{h}_k|^2}{\sum_{j\neq k}P_j|\widehat{\mathbf{v}}_k^H\mathbf{h}_j|^2 + N_0\widehat{\mathbf{v}}_k^H\widehat{\mathbf{v}}_k},
	\end{aligned}
	\end{equation}
	where $\widehat{\mathbf{v}}_k$ is the $k$th column of $\widehat{\mathbf{V}}$.
	\par 
	Notice that when the distance threshold $d_0$ is small, the matrix $\widehat{\mathbf{H}}$ can be very sparse, leading to a significant reduction in channel estimation overhead and processing complexity. The key question is: how small $d_0$ can be or how sparse the channel matrix can be without significantly affecting the system performance. In other words, how should we set $ d_0$ so that $\widehat{\text{SINR}}_k$ is not much lower than $\text{SINR}_k$ in (\ref{eqn:sinr}). This question will be answered in the next subsection.
		\vspace{-0.4cm}
	\subsection{Distance Threshold Analysis}
	In this subsection, we show how to set the distance threshold $d_0$ if a high percentage of full SINR is to be achieved. Specifically, we wish to set $d_0$, such that the SINR ratio, defined as
	\begin{equation}
	\small
	\rho(d_0)=\frac{\mathrm{E}[\widehat{\text{SINR}}_k(d_0)]}{\mathrm{E}[\text{SINR}_k]}
	\end{equation}
	is larger than a prescribed level $\rho^*$, where the expectation $\mathrm{E}(\cdot)$ is taken over $\mathbf{H}$, with randomness including both path loss and Rayleigh fading.
	\par
	In the following, we endeavour to derive a closed-form expression of $d_0$ as a function of the target SINR ratio $\rho^*$. Let us first introduce two approximations that make our analysis tractable.
	\par 
{\textbf{{Approximation 1:}}} The distances $d_{n,k}$, for all $n,k$ are mutually independent.
\par 
	As shown in Fig.~\ref{indep1}, we plot the SINR ratio for systems with and without Approximation 1. The system area is assumed to be a circle with radius $2.5$ km. The figure shows that the gap between the SINR ratio is very small, which validates the independence approximation.
	\par 
	{\textbf{{Approximation 2:}}} Conditioning on the distance threshold $d_0$, the matrices $\widehat{\mathbf{H}}$ and $\widetilde{\mathbf{H}}$ are mutually independent.
	\par 
	Note that $\mathrm{E}[\widehat{\mathbf{H}}\widetilde{\mathbf{H}}^H]=\mathrm{E}[\widetilde{\mathbf{H}}\widehat{\mathbf{H}}^H]=\mathbf{0}$, which means that $\widehat{\mathbf{H}}$ and $\widetilde{\mathbf{H}}$ are uncorrelated. With the independence approximation, the equality 
	\begin{equation}
	\small
		\mathrm{E}_{\mathbf{H}}[\widehat{\text{SINR}}_k(d_0)]=\mathrm{E}_{\widehat{\mathbf{H}}}\left[\mathrm{E}_{\widetilde{\mathbf{H}}}[\widehat{\text{SINR}}_k(d_0)]\right]
\end{equation}		
holds. This approximation will be verified in our numerical results in Fig.~\ref{s1}, which shows that the gap between the simulated SINR ratio and the lower bound of $\rho(d_0)$ derived based on this approximation is small.
\begin{figure}[!h]
	\centering
	{\includegraphics[width=0.46\textwidth]{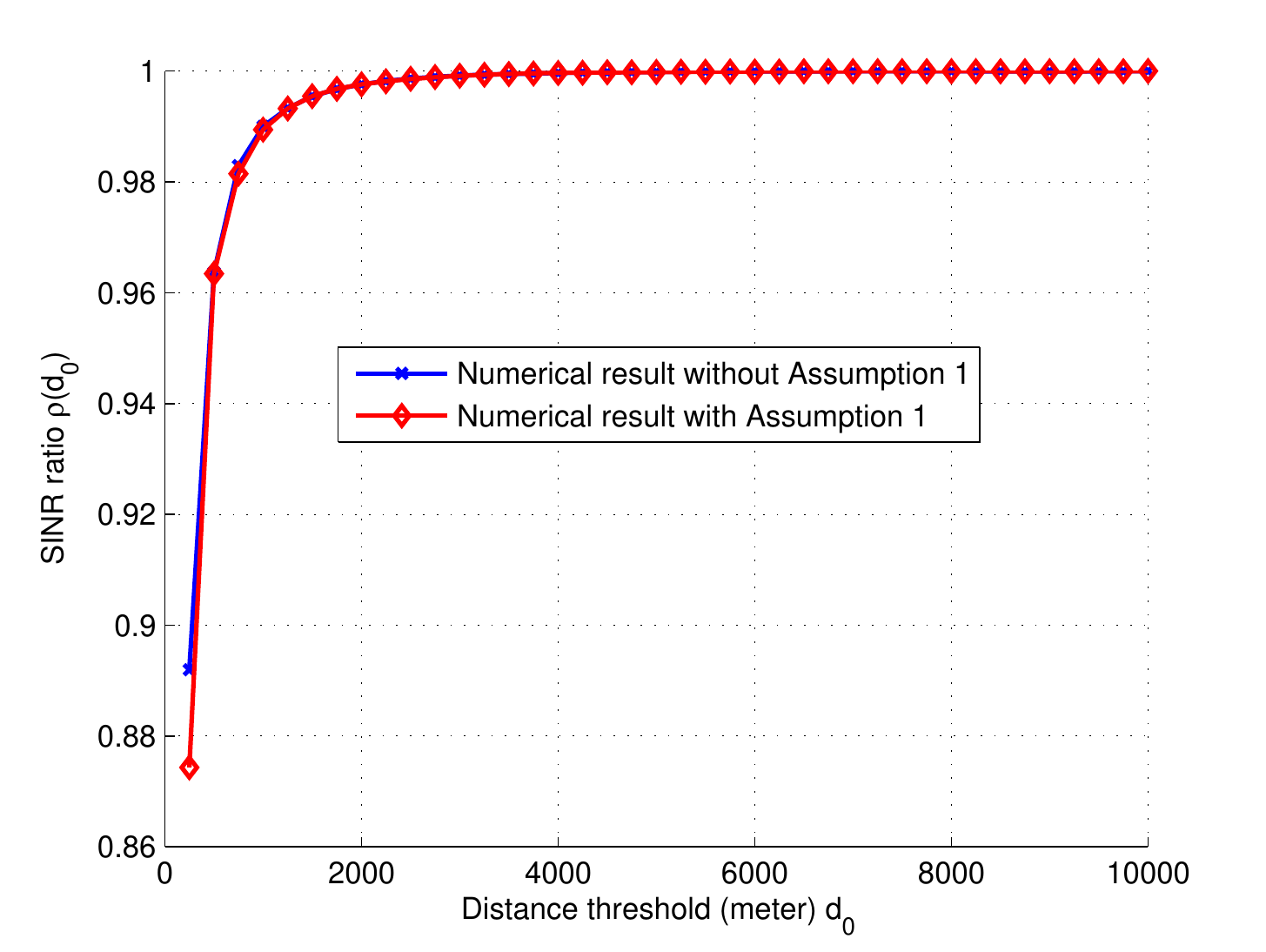}}
	\caption{Average SINR ratio vs distance threshold when $N=1000, K=600$.}
	\label{indep1}
	\end{figure}
	\par 
	Based on these two aproximations, we see that $\mathbf{\Gamma}=N_1\mathbf{I}$, where $N_1=\mathrm{E}[\sum_{j=1}^{K}P_j|\widetilde{h}_{n,j}|^2]$ for arbitrary RRH $n$. We can now derive a lower bound of $\rho(d_0)$ as shown in the following Theorem 1.
	\begin{theorem}
	Given a distance threshold $d_0$, a lower bound of SINR ratio $\rho(d_0)$ is given by
	\begin{equation}
	\small
	\rho(d_0)\geq \underline{\rho(d_0)} \triangleq \frac{\widehat{\mu}N_0}{\mu\left(\left(\mu-\widehat{\mu}\right)\sum_{j \neq k}P_k+N_0\right)},\label{lowerboundsinr1}
	\end{equation}
	where $\widehat{\mu}
	=\int_{x=0}^{d_0}x^{-\alpha}f(x)\text{d}x$ and $\mu=\int_{x=0}^{\infty}x^{-\alpha}f(x)\text{d}x$ respectively, and $f(x)$ is the probability density function of the distance between RRHs and users. 
	\par 
	When each user transmits at the same amount of power $P$, the lower bound is simplified as
	\begin{equation}
	\small
	\underline{\rho(d_0)} = \frac{\widehat{\mu}N_0}{\mu\left(P\left(\mu-\widehat{\mu}\right)\left(K-1\right)+N_0\right)}.\label{lowerboundsinr2}
	\end{equation}
	\end{theorem}	
	\begin{proof}
	See Appendix.
	\end{proof}
	\par 
	We notice that $\underline{\rho(d_0)}$ depends on the probability distribution of the distances between mobile users and RRHs. In \cite{DDistance}, distance distributions are derived for different network area shapes, such as circle, square and rectangle. Take, for example, a circular network area with radius $r$. In this case, the distance distribution between two random points is \cite{DDistance}
	\begin{equation}
	\small
	\begin{aligned}
	f(x,r)
	=&
	\begin{cases}
	\int_{0}^{r_0}\frac{2x}{r^2}\left(\frac{2}{\pi}\arccos\left(\frac{y}{2r}\right)-\frac{y}{\pi r}\sqrt{1-\frac{y^2}{4r^2}}\right)\text{d}y,&x=r_0,
	\\
	\frac{2x}{r^2}\left(\frac{2}{\pi}\arccos\left(\frac{x}{2r}\right)-\frac{x}{\pi r}\sqrt{1-\frac{x^2}{4r^2}}\right),&r_0<x<2r,
	\end{cases}\label{pdfofd}
	\end{aligned}
	\end{equation}
	where $r_0$ is the minimum distance between RRHs and users. 
	\par 
	When the network radius $r$ becomes very large, (\ref{pdfofd}) can be approximated as
	\begin{equation}
	\small
	\hat{f}(x,r)=\begin{cases}
	\frac{r_0^2}{r^2},&x=r_0,
	\\
	\frac{2}{r^2}x,&r_0<x<r.
	\end{cases}\label{apdfofd}
	\end{equation}
	Substituting (\ref{pdfofd}) or (\ref{apdfofd}) into (\ref{lowerboundsinr2}), we obtain the relation between $d_0$ and the SINR requirement $\rho^*$:
	\begin{theorem}
	When $d_0$ is the solution of 
	\begin{equation}
	\small
	\begin{aligned}
	N_0\int_{x=0}^{d_0}x^{-\alpha}f(x,r)\text{d}x
	=&\rho^*\left(P\left(K-1\right)\int_{x=d_0}^{\infty}x^{-\alpha}f(x,r)\text{d}x\right)\int_{x=0}^{\infty}x^{-\alpha}f(x,r)\text{d}x,
	\end{aligned}\label{eqn:d0}
	\end{equation}
	where $f(x,r)$ is given in (\ref{pdfofd}), an SINR ratio no smaller than $\rho^*$ can be achieved. 
	\par 
	When the network size is very large (i.e., $r \gg r_0$), the solution to (\ref{eqn:d0}) can be approximated as 
	\begin{equation}
	\small
	\begin{aligned}
	&d_0 = \left( r^{2-\alpha}+\frac{(\alpha r_0^{2-\alpha}-2 r^{2-\alpha})(1-\rho^*)N_0}{2N_0+\frac{2\rho^* (\alpha r_0^{2-\alpha}-2r^{2-\alpha})(K-1)P}{(\alpha-2)r^2}}\right)^{-\frac{1}{\alpha-2}}.
	\end{aligned}\label{eqn:dvsrho1}
	\end{equation}
	\par 
	Particularly, when the network size goes to infinity (i.e., $r \rightarrow \infty$), $d_0$ can be further simplified as
	\begin{equation}
	\small
	\begin{aligned}
	&d_0 = \left(\frac{2N_0(\alpha-2) +2\alpha r_0^{\alpha-2}\rho^* \pi \beta_KP}{\alpha r_0^{2-\alpha} N_0(1-\rho^*)(\alpha-2)}\right)^{\frac{1}{\alpha-2}}.
	\end{aligned}\label{eqn:dvsrho2}					
	\end{equation}
	\end{theorem}
	\subsection{Verifications and Discussions}
	In this subsection, we first verify our analysis through numerical simulations. We then illustrate the effect of SINR ratio requirement on the choice of the distance threshold. Finally, we illustrate the sparsity of the channel matrix and discuss the possibility of reducing estimation overhead based on the sparsified matrix. Unless stated otherwise, we assume that the minimum distance between RRHs and users is $1$ meter, the path loss exponent is 3.7, and the average transmit SNR at the user side equals to $80$dB. That is $\frac{P}{N_0}=80$dB.
	\begin{figure}[!h]
	\centering
	{\includegraphics[width=0.46\textwidth]{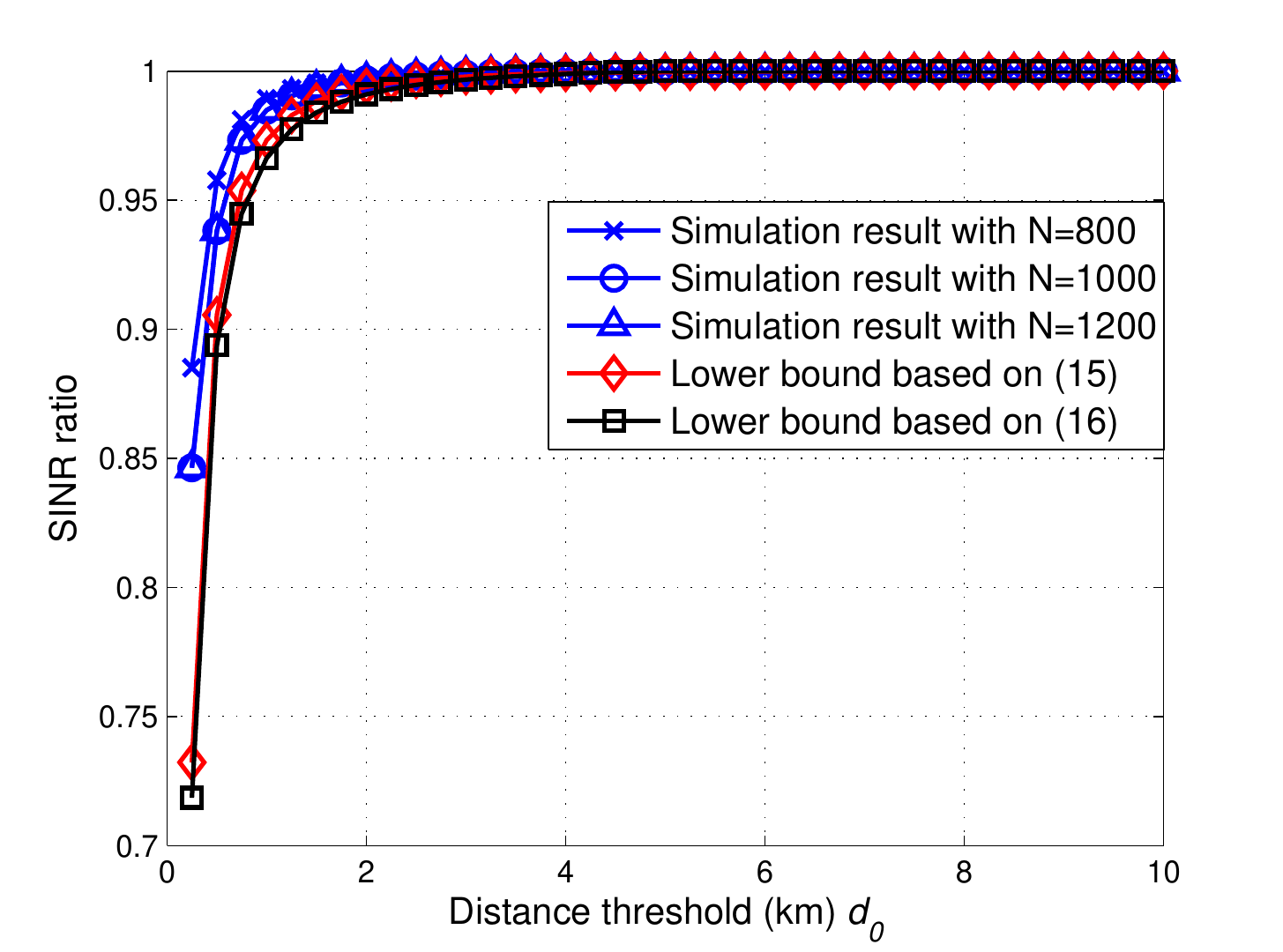}}
	\caption{Average SINR ratio vs distance threshold when $K=1000, r=5$km.}\label{s1}
	\vspace{-0.5cm}
	\end{figure}
	\begin{figure}[!h]
	\centering
	{\includegraphics[width=0.46\textwidth]{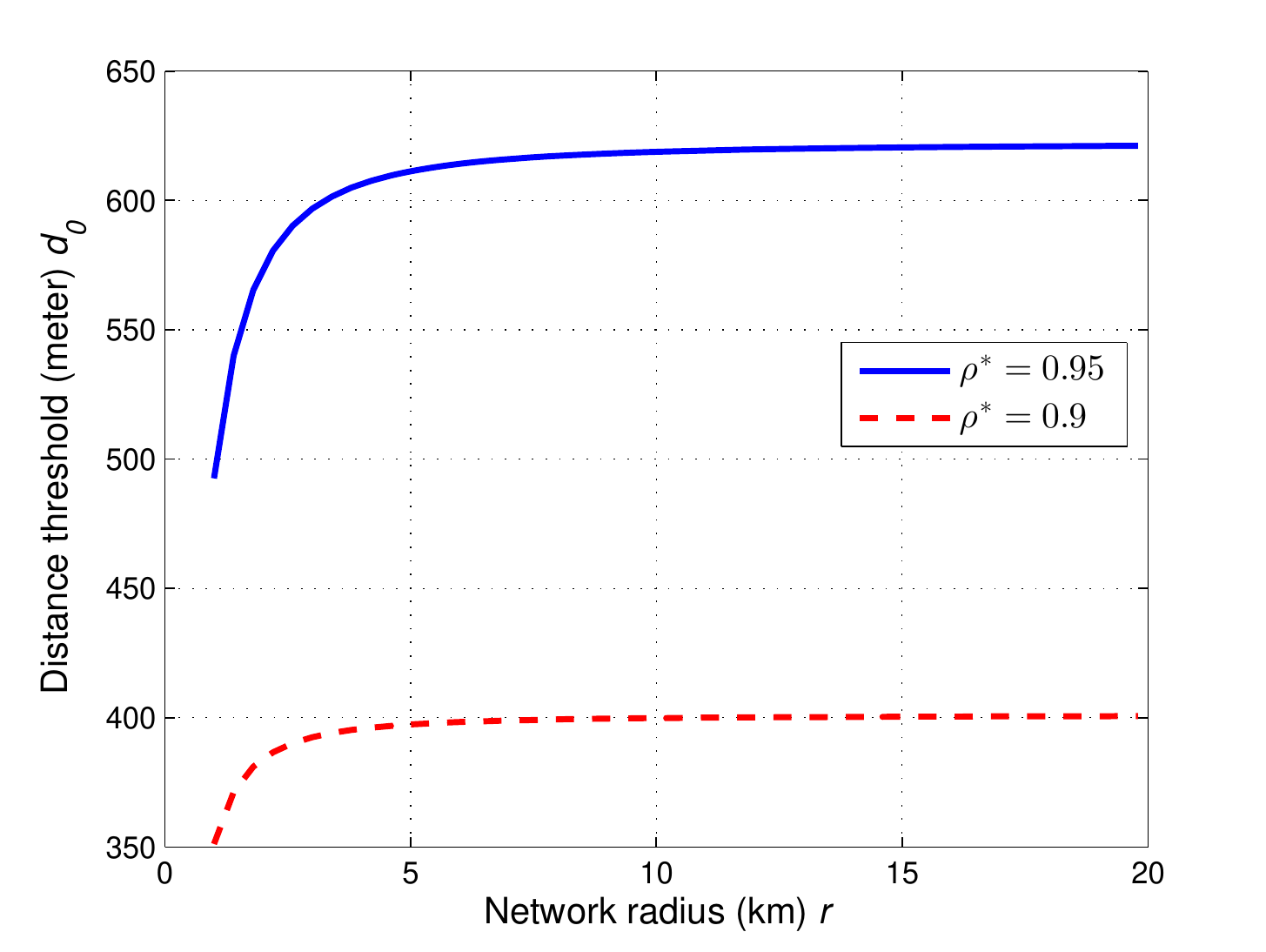}}
	\caption{Distance threshold $d_0$ vs area radius $r$ when the user density $\beta_K=8$/km$^2$.}\label{s3}
	\end{figure}
	\par
	\subsubsection{Verification of Theorem 1 and 2}
	\par 
	Fig.~\ref{s1} plots the SINR ratio against the distance threshold, when $K=1000$ and $r=5$ km.  The simulated SINR ratio with different numbers of RRHs, $N$, are plotted as the blue curves and $\underline{\rho(d_0)}$ derived based on the distributions in (\ref{pdfofd}) and (\ref{apdfofd}) are plotted as the red and black curves, respectively. We can see that the gap between the lower bound based on (\ref{pdfofd}) and that based on (\ref{apdfofd}) is negligible, which means that the approximation of distance distribution is reasonable. Moreover, we notice that even though the simulated ratios vary with $N$, the lower bounds derived based on Theorem 1 and (\ref{pdfofd}), (\ref{apdfofd}) remain unchanged for different $N$. 
	
	\par 
	In Fig.~\ref{s3}, we show that the distance threshold converges to a constant when the network radius $r$ becomes large, as predicted in (\ref{eqn:dvsrho2}). Here, the user density is $\beta_K=8/\text{km}^2$, and the SINR ratio requirement is set to $\rho^*=0.95$ and $\rho^*=0.9$, respectively. As expected, the distance threshold converges quickly to a constant when the network radius increases. Indeed, the convergence is observed even when the network radius is as small as $5$ km for both $\rho^*=0.9$ and $\rho^*=0.95$.
	\begin{figure}[!h]
	\centering
	{\includegraphics[width=0.46\textwidth]{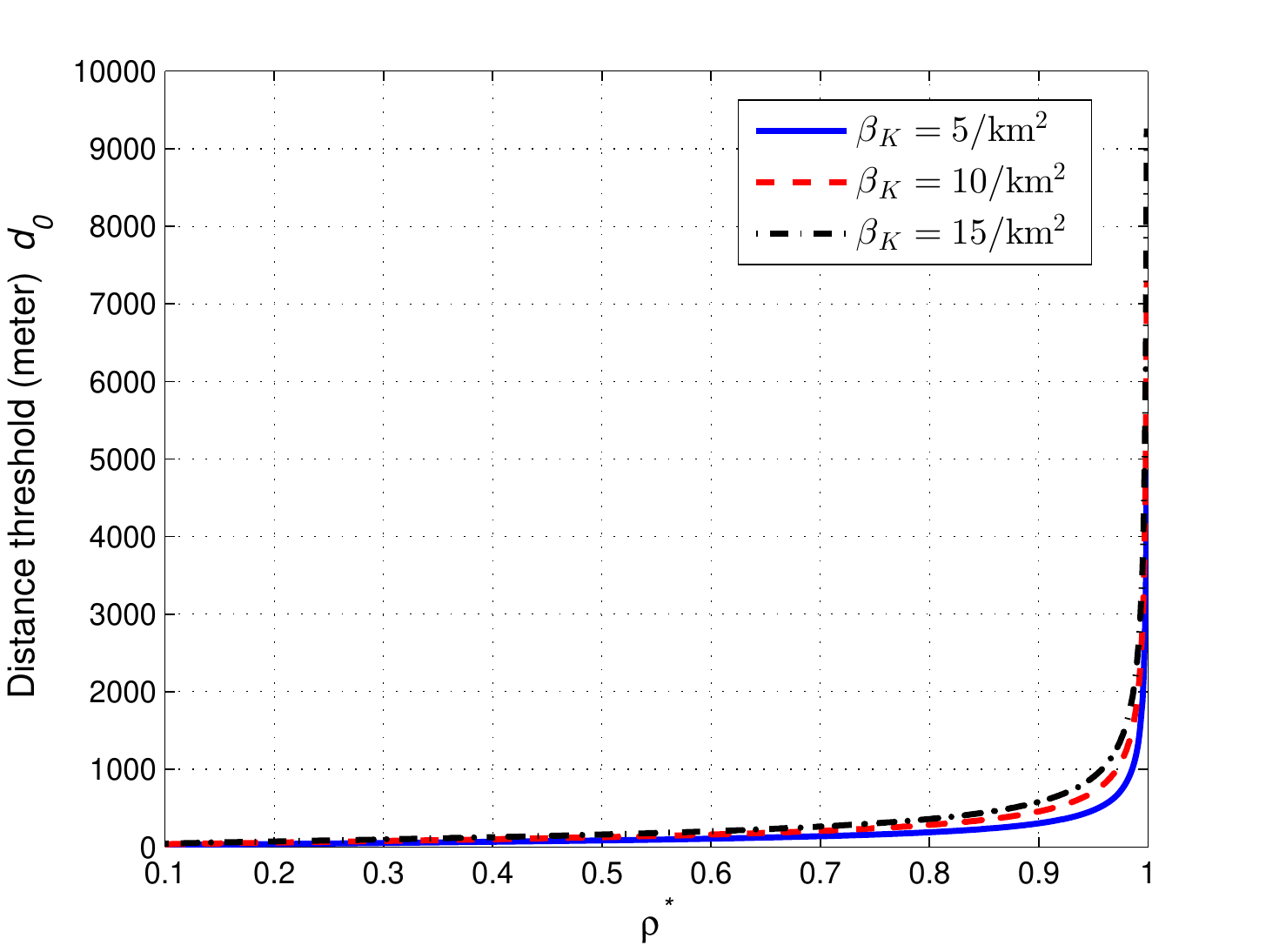}}
	\caption{Distance threshold vs SINR ratio.}\label{s4}
	\end{figure}
	\par 
	\subsubsection{SINR Loss versus Distance Threshold}
	\par 
	We then discuss the effect of the SINR requirement $\rho^*$ on the distance threshold. In Fig.~\ref{s4}, we plot the distance thresholds against $\rho^*$, when user density $\beta_K=5$, $10$, and $15/\text{km}^2$, respectively. The network radius is assumed to be very large. We can see that the distance threshold remains very small for a wide range of $\rho^*$, i.e., when $\rho^*$ is smaller than 0.95. There is a sharp increase in $d_0$ when $\rho^*$ approaches 1. This implies an interesting tradeoff: if full SINR is to be achieved, we do need to process the full channel matrix $\mathbf{H}$ at the cost of high complexity when the network size is large. On the other hand, if a small percentage of SINR degradation can be tolerated, the channel matrix can be significantly sparsified, leading to low-complexity detection algorithms. We emphasize that the SINR degradation may not imply a loss in the system capacity. This is because the overhead of estimating the full channel matrix can easily outweigh the SINR gain. A little compromise in SINR (say, reducing from $100\%$ to $95\%$) may yield a higher system capacity eventually.
	\begin{table}[htbp]
	\small
	\label{tab:parameter}
	        \caption{Percentage of non-zero entries in the channel matrix with $\beta_K=10/\text{km}^2, \frac{P}{N_0}=80\text{\textnormal{d}B}$ and $\rho^*=0.95$}
	        \centering
	        \begin{tabular}{c | c c c c}
	        \toprule
	   		$r$ (km)& $5$ & $10$ & $15$ & $20$\\
	        \midrule 
	        $d_0$ (meter)& $694$ & $705$ & $707$ & $708$\\
	        \midrule
	Percentage of non-zero entries ($\%$)& $1.93$ & $0.50$ & $0.20$ & $0.13$\\
	         \bottomrule
	        \end{tabular}
	        \label{table:t1}
	\end{table}
	\begin{table}[htbp]
	\small
	\label{tab:parameter}
	        \caption{Percentage of non-zero entries in the channel matrix with $\beta_K=10/\text{km}^2, \frac{P}{N_0}=80\text{\textnormal{d}B}$ and $r=10$km}
	        \centering
	        \begin{tabular}{c | c c c c}
	        \toprule
	   		$\rho^*$ & $0.90$ & $0.93$ & $0.96$ & $0.99$\\
	        \midrule 
	        $d_0$ (meter)& $456$ & $572$ & $807$ & $1812$\\
	        \midrule
	Percentage of non-zero entries ($\%$)& $0.21$ & $0.33$ & $0.65$ & $3.28$\\
	         \bottomrule
	        \end{tabular}
	        \label{table:t2}
	\end{table}     
	\par 
	\subsubsection{Sparsity of $\widehat{\mathbf{H}}$}
	\par 
	As seen from (\ref{eqn:dvsrho1}) and (\ref{eqn:dvsrho2}), for a given $\rho^*$, the distance threshold $d_0$ converges to a constant when the network radius $r$ goes to infinity. Since the average number of non-zero channel coefficients each RRH is approximately $\pi d_0^2 \beta_K$, the convergence of $d_0$ implies that the number of non-zero entries per row or per column in $\widehat{\mathbf{H}}$ does not scale with the network radius $r$ in a large C-RAN. Moreover, the percentage of non-zero entries in $\widehat{\mathbf{H}}$ is approximately $\frac{d_0^2}{r^2}$, which can be very small when $r$ is large. In Table \ref{table:t1}, we list both $d_0$ and the corresponding percentage of non-zero entries in matrix $\widehat{\mathbf{H}}$ for various network sizes, with $\beta_K=10/\text{km}^2$ and $\rho^*=0.95$. It can be seen that, when $r$ is large, $d_0$ does not change much with the network radius $r$. Moreover, only a small percentage of entries (say $2\% \sim 0.13\%$) in $\widehat{\mathbf{H}}$ are non-zero  for all values of $r$ considered in Table \ref{table:t1}. In other words, each RRH only needs to estimate the CSI of a small number of users closest to this RRH. The channel estimation overhead can be significantly reduced. If a larger SINR loss can be tolerated, the amount of CSI needed can be further reduced as shown in Table \ref{table:t2}, which lists the percentages of non-zero entries in $\widehat{\mathbf{H}}$ for different $\rho^*$, with $\beta_K=10/\text{km}^2$ and $r=10$km. We see that the percentage of non-zero entries can be reduced from $3.28\%$ to $0.21\%$ by allowing a drop of the the SINR performance from $99\%$ to $90\%$.
	\begin{remark}
	As we can see from the figures, close-to-$100\%$ SINR is achievable when the channel matrix is reasonably sparsified. Notice that sparsifying the channel matrix leads to a significant reduction in channel estimation overhead. This is because we only need to estimate the small scale fadings of the matrix entries that have not been discarded. Therefore,  matrix sparsification may effectively lead to a higher system capacity due to the reduction of channel estimation overhead, despite a small decrease in SINR.
	\end{remark}
	\vspace{-0.3cm}
\section{Single-layer Dynamic Nested Clustering}
With the sparsified channel matrix, we now proceed to present the single-layer DNC algorithm in this section. As shown in (\ref{eqn:v_hat}), to estimation $\mathbf{x}$ is to calculate $\mathbf P^{\frac{1}{2}}\widehat{\mathbf{H}}^H\widehat{\mathbf{A}}^{-1}\mathbf{y}$. Note that the computational complexity is $O(N^3)$
which is dominated by calculating $\widehat{\mathbf{A}}^{-1}$. This is because the sparse matrix $\widehat{\mathbf{H}}$ only contains a constant number of non-zeros per column, and so does $\mathbf P^{\frac{1}{2}}\widehat{\mathbf{H}}^H$, when the network area goes to infinity. Suppose the average number of non-zero entries in each column of $\widehat{\mathbf{H}}$ is $c$. The computational complexity of multiplying $\mathbf P^{\frac{1}{2}}{\widehat{\mathbf H}}^H$ and $\widehat{\mathbf{A}}^{-1}\mathbf{y} \in \mathcal C^{N \times 1}$ is only $O(cN)$ and is much smaller than $O(N^3)$, i.e., the computational complexity of inverting $\widehat{\mathbf{A}}$. Therefore, we focus on reducing the computational complexity of calculating $\widehat{\mathbf{A}}^{-1}\mathbf{y}$, which is equivalent to solving for $\omega$ in the equation:
	\begin{equation}	
	\small
	\widehat{\mathbf{A}}\mathbf{\omega}=\mathbf{y}.
	\label{eqn:dbbd1}	
	\end{equation}
\par 
It is obvious that the matrix $\widehat{\mathbf{A}}$ is a sparse matrix. The sparse linear equations, i.e., (\ref{eqn:dbbd1}), have been well studied in multiple areas, such as the numerical linear algebra, graph theory, \textit{etc}. However, most of the existing works proposed iterative algorithms, whose accuracy and convergence cannot be guaranteed \cite{saaditerative}. Here, we propose a direct algorithm to obtain an accurate solution of (\ref{eqn:dbbd1}). Before go into the details of the algorithm, we first explain the physical meanings of the entries in $\widehat{\mathbf{A}}$ as follows. According to the threshold-based channel matrix sparsification approach, the $(n,k)$th entry of channel matrix $\widehat{\mathbf{H}}$ is non-zero only when the $k$th user is in the service area of RRH $n$, i.e., a circular area with radius $d_0$ centered around RRH $n$. Consequently, from the definition of $\widehat{\mathbf{A}}$ in (\ref{eqn:v_hat}), the $(n_1, n_2)$th entry in $\widehat{\mathbf{A}}$ is non-zero only when the service areas of RRH $n_1$ and $n_2$ overlap, and there is at least one user in the overlapping area.
\par 
Consider an ideal case where the whole set of RRHs can be divided into disjoint clusters. While the RRHs within one cluster have overlapping service areas, those from different clusters do not serve the same user(s). In this case, the matrix $\widehat{\mathbf{A}}$ becomes block diagonal with each block corresponding to one cluster. Then, the complexity of calculating $\widehat{\mathbf {A}}^{-1}$ reduces from $O(N^3)$ to $O(n_i^3)$, where $n_i$ is the number of RRHs in a cluster. Note that $n_i$ is typically much smaller than $N$, i.e., the total number of RRHs in a C-RAN.
\par
In reality, however, adjacent clusters interact and interfere with each other. Particularly, the service areas of the RRHs in adjacent clusters are likely to overlap. Traditional clustering algorithms \cite{interference2012akoum,lee2014spectral} usually ignore such overlapping, resulting in a noticeable performance degradation. In what follows, we show that by properly labeling the RRHs, matrix $\widehat{\mathbf{A}}$ can be transformed to a DBBD matrix, where the borders capture the overlaps between clusters. Then, later in Subsection IV-B, the DNC algorithm that enables parallel computation is presented.
\vspace{-0.3cm}
\subsection{RRH Labelling Algorithm}
To start with, we give the definition of a Hermitian DBBD matrix as follows:
\vspace{-0.1cm}
\begin{definition}
	A matrix $\mathbf{A}$ is said to be a Hermitian DBBD matrix if it is in the following form
	\begin{equation}
	\small
	\mathbf{A}=
	\begin{bmatrix}
     \mathbf{A}_{1,1} &  & & &\mathbf{A}_{c1}^H \\[0.3em]
       &\mathbf{A}_{2,2} & & &\mathbf{A}_{c2}^H\\[0.3em]
        & &\cdots & &\vdots\\[0.3em]
         & & &\mathbf{A}_{m,m} &\mathbf{A}_{cm}^H\\[0.3em]
       \mathbf{A}_{c1}& \mathbf{A}_{c2}&\cdots&\mathbf{A}_{cm} &\mathbf{A}_{c}
     \end{bmatrix},
	\end{equation}
	where the diagonal blocks $\mathbf{A}_{ii}$ are $n_i\times n_i$ Hermitian matrices, the border blocks $\mathbf{A}_{ci}$ are $n_c \times n_i$ matrices, and the cut-node block $\mathbf{A}_c$ is an $n_c \times n_c$ Hermitian matrix.
	\end{definition}
	\begin{figure}[!h]
	\centering
	{\includegraphics[width=0.63\textwidth]{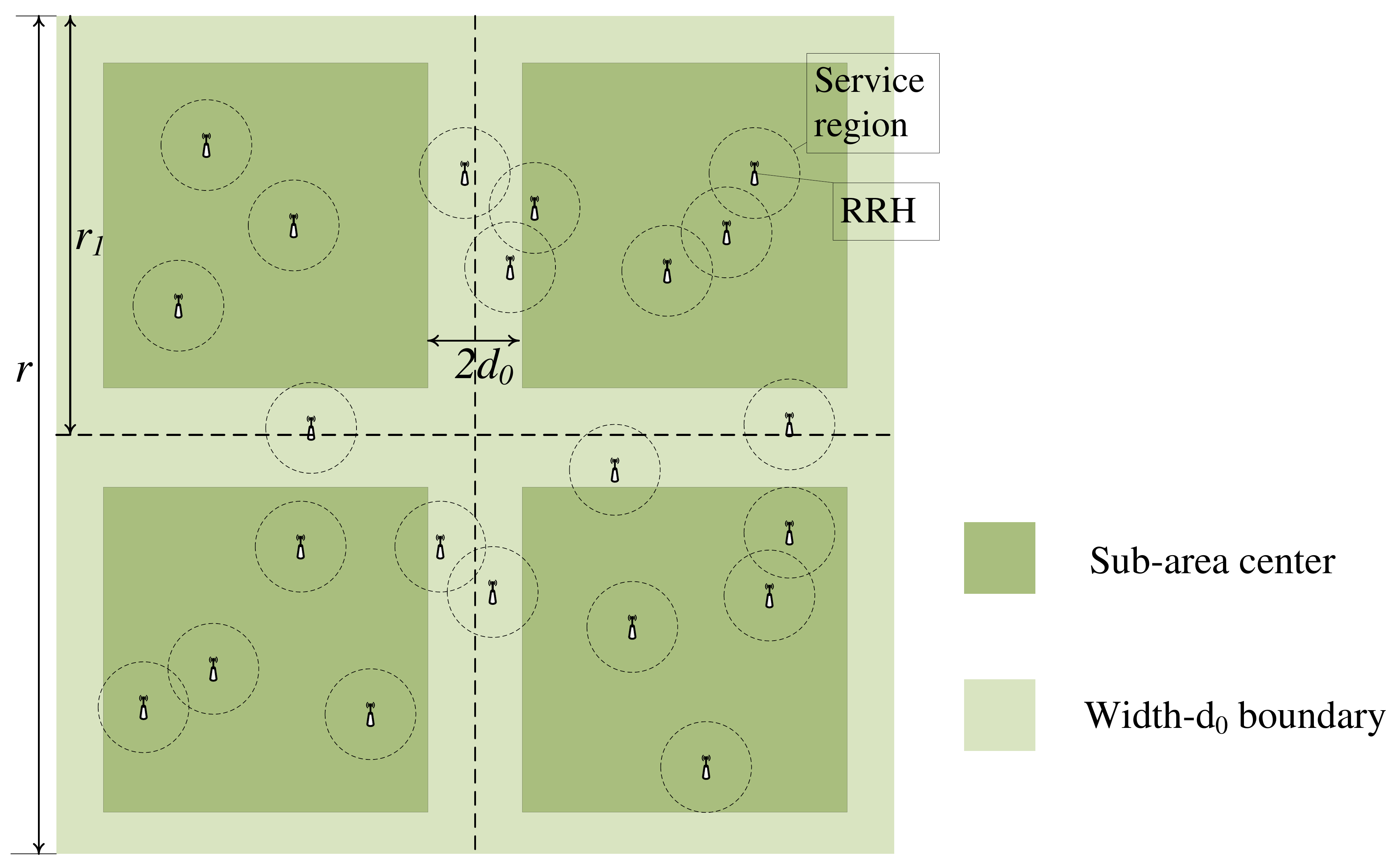}}
	\caption{Geographical RRH grouping in C-RAN.}\label{location1}
\end{figure}
\par 
	We divide the entire C-RAN area into disjoint sub-areas as illustrated in Fig.~\ref{location1}. We then separate each sub-area into a width-$d_0$ boundary and a sub-area center which are colored by light-green and dark-green respectively. Here, $d_0$ is the distance threshold used in the channel sparification. We see that RRHs in a sub-area center do not have overlapping service region with RRHs in other sub-areas. Only RRHs in the width-$d_0$ boundary may have overlapping service region with the RRHs in adjacent sub-areas. This implies that matrix $\widehat{\mathbf{A}}$ can be transformed to a DBBD matrix with each diagonal block corresponding to RRHs in a sub-area center and the cut-node block corresponding to RRHs in the width-$d_0$ boundaries. The border blocks of $\widehat{\mathbf{A}}$ capture the interaction between different clusters due to interference. 
	\par 
	
	\begin{algorithm}
\caption{RRH Labelling Algorithm}
\begin{algorithmic}[1]
\REQUIRE  
{$a_x$, $a_y$, $d_0$, $r_1$, $l_n$, $\forall n$ }
\ENSURE 
{$b(n)$, $\forall n$}
\STATE Set $m_x = \lceil \frac{a_x}{r_1}\rceil$, $m_y = \lceil \frac{a_y}{r_1}\rceil$ and $\mathcal C_i=\Phi, \forall i \in \{1,2, \cdots, m_xm_y+1\}$
\FOR{$n=1$ to $N$}
\STATE
Setting $i =\lceil \frac{lx_n}{r_1}\rceil$\\
$j =\lceil \frac{ly_n}{r_1}\rceil$
\IF{$(i-1)r_1+d_0\leq lx_n \leq ir_1-d_0 \text{ AND } (j-1)r_1+d_0\leq ly_n \leq jr_1-d_0$}
\STATE{
$\mathcal C_{(i-1)m_x+j} \leftarrow \{\mathcal C_{(i-1)m_x+j}, n\}$
}
\ELSE
\STATE{
$\mathcal C_{m_xm_y+1} \leftarrow \{\mathcal C_{m_xm_y+1}, n\}$
}
\ENDIF
\ENDFOR
\STATE Set $j=1$
\FOR{$i=1$ to $m_xm_y+1$}
\FOR{$n=1$ to $N$}
\IF{$n \in \mathcal C_i$}
\STATE{Label RRH $n$ by $j$: $b(n)\leftarrow j$}
\STATE{$j \leftarrow j+1$}
\ENDIF
\ENDFOR
\ENDFOR
\label{algorithm:a1}
\end{algorithmic} 
\end{algorithm}	
	
Denote the coordinates of an arbitrary RRH $n$ as follows:
	\begin{equation}
	\small
	l_n=(lx_n,ly_n),
\end{equation}		
	where $lx_n \in [0, a_x], ly_n \in [0, a_y]$, $a_x$ and $a_y$ are the side lengths of the whole network. The RRH labelling algorithm is given in Algorithm 1, where $b(n)$ is the label of RRH $n$. We first divide the overall network into disjoint squares with side length $r_1$, and group the RRHs into center clusters or the boundary cluster according to their locations in steps 2 to 9. Then, the RRHs are numbered based on the cluster they belong to, as shown in steps 10 to 18. After numbering all the RRHs, we organize the matrix $\widehat{\mathbf{A}}$ and the signal vector $\mathbf{y}$ in the ascending order of the RRHs' numbers. For example, the first row of $\widehat{\mathbf{A}}$ corresponds to the RRH with label $b(n)=1$. That is, the first row captures the interaction between that RRH with label $b(n)=1$ and all RRHs in the network. The matrix $\widehat{\mathbf{A}}$ now becomes a DBBD matrix.
\par
 \begin{remark} In this paper, we assume that the computational complexity of MMSE detection grows cubically with the number of RRHs. As such, RRH labelling in Algorithm 1 is only based on the locations of RRHs, but not on those of the users. However, notice that the actual computational complexity of MMSE detection is $O(M^3)$ instead of $O(N^3)$, where $M=\text{min}(N,K)$. This is because if $K< N$, the estimation of $\mathbf{x}$ can be obtained by calculating $\left(\mathbf P^{\frac{1}{2}}\widehat{\mathbf{H}}^H\widehat{\mathbf{H}}\mathbf P^{\frac{1}{2}}\right)^{-1}\mathbf P^{\frac{1}{2}}\widehat{\mathbf{H}}^H\mathbf{y}$. Now the complexity of the matrix inversion becomes $O(K^3)$. In our paper, we just take the case that $N\leq K$ as an example. The algorithm can be adapted to the case $N>K$ by applying the labelling algorithm to users instead of RRHs.
 \end{remark} 
\vspace{-0.4cm}
\subsection{Single-Layer DNC with Parallel Computing}
\begin{figure*}[!t]
	\centering
	{\includegraphics[width=0.85\textwidth]{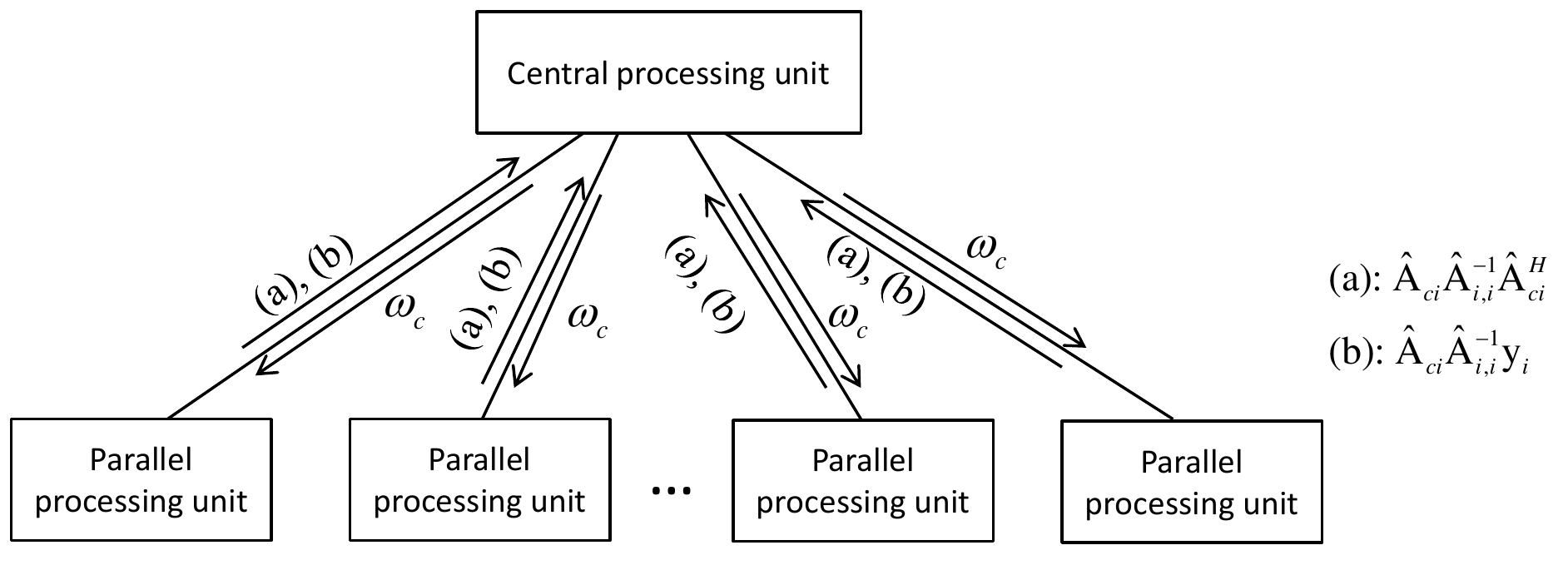}}
	\caption{Architecture of parallel computing of the single-layer DNC alogrithm.}\label{parallel}	
\end{figure*}
With $\widehat{\mathbf{A}}$ converted into a DBBD matrix, we are now ready to present the DNC algorithm. In particular, the diagonal blocks of $\widehat{\mathbf{A}}$ can be processed in parallel, leading to a significant reduction in computation time. \par 
	Suppose that the $N\times N$ DBBD matrix $\widehat{\mathbf{A}}$ has $m_1$ diagonal blocks. Then, the linear equation (\ref{eqn:dbbd1}) becomes
	\begin{equation}
	\small
\begin{bmatrix}
     \widehat{\mathbf{A}}_{1,1} &  & & &\widehat{\mathbf{A}}_{c1}^H \\[0.3em]
       &\widehat{\mathbf{A}}_{2,2} & & &\widehat{\mathbf{A}}_{c2}^H\\[0.3em]
        & &\cdots & &\vdots\\[0.3em]
         & & &\widehat{\mathbf{A}}_{m_1,m_1} &\widehat{\mathbf{A}}_{cm_1}^H\\[0.3em]
       \widehat{\mathbf{A}}_{c1}& \widehat{\mathbf{A}}_{c2}&\cdots&\widehat{\mathbf{A}}_{cm_1} &\widehat{\mathbf{A}}_{c}
     \end{bmatrix}
     \begin{bmatrix}
     \mathbf \omega_1\\[0.3em]
     \mathbf \omega_2\\[0.3em]
     \vdots\\[0.3em]
     \mathbf \omega_{m_1}\\[0.3em]
     \mathbf \omega_c
     \end{bmatrix}
     =
     \begin{bmatrix}
     \mathbf y_1\\[0.3em]
     \mathbf y_2\\[0.3em]
     \vdots\\[0.3em]
     \mathbf y_{m_1}\\[0.3em]
     \mathbf y_c
     \end{bmatrix},\label{eqn:dbbd_extend}
\end{equation}
where the $n_i\times 1$ vectors $\omega_i$ and $\mathbf y_i$ are sub-vectors of $\omega$ and $\mathbf y$, respectively. Likewise, $\mathbf \omega_c$ and $\mathbf y_c$ are $n_c \times 1$ sub-vectors.
\par 
The solution to the above equation is given by
\begin{equation}
\small
\mathbf \omega_c=\left(\widehat{\mathbf{A}}_c-\sum_{i=1}^{i=m_1}\widehat{\mathbf{A}}_{ci}\widehat{\mathbf{A}}_{i,i}^{-1}\widehat{\mathbf{A}}_{ci}^H\right)^{-1}\left(\mathbf y_c-\sum_{i=1}^{i=m_1}\widehat{\mathbf{A}}_{ci}\widehat{\mathbf{A}}_{i,i}^{-1}\mathbf y_i\right),\label{eqn:x1}
\end{equation}
and 
\begin{equation}
\small
\mathbf \omega_i=\widehat{\mathbf{A}}_{i,i}^{-1}\left(\mathbf y_i-\widehat{\mathbf{A}}_{ci}^H \mathbf \omega_c\right),\label{eqn:x2}
\end{equation}
 for all $i \in \{1,2,\cdots, m_1\}$.
\begin{table*}[!t]
\small
	\label{tab:parameter}
	        \caption{computation time of each step in equation (\ref{eqn:x1}) and (\ref{eqn:x2})}
	        \centering
	        \begin{tabular}{c  c  c  c }
	        \toprule
	   		step & operation & complexity/operation & total number of operations\\
	        \midrule 
	        $1$ & $\widehat{\mathbf{A}}_{i,i}^{-1}$ & $O(N_{d,1}^3)$ & $m_1$\\
	         \midrule 
	          \multicolumn{1}{c}{\multirow {2}{*}{$2$}} & \multicolumn{1}{c}{$\widehat{\mathbf{A}}_{ci}\widehat{\mathbf{A}}_{i,i}^{-1}\widehat{\mathbf{A}}_{ci}^H$} & \multicolumn{1}{c}{\multirow {2}{*}{$O(\frac{L_1}{m_1}N_{b,1}N_{d,1})$}} & \multicolumn{1}{c}{\multirow {2}{*}{$m_1$}}\\
         \multicolumn{1}{c}{} & \multicolumn{1}{c}{$\widehat{\mathbf{A}}_{ci}\widehat{\mathbf{A}}_{i,i}^{-1}\mathbf y_i$ } & \multicolumn{1}{c}{} & \multicolumn{1}{c}{} \\    
	        \midrule 
	          \multicolumn{1}{c}{\multirow {2}{*}{$3$}} & \multicolumn{1}{c}{$\widehat{\mathbf{A}}_c-\sum_{i=1}^{i=m_1}\widehat{\mathbf{A}}_{ci}\widehat{\mathbf{A}}_{i,i}^{-1}\widehat{\mathbf{A}}_{ci}^H$} & \multicolumn{1}{c}{\multirow {2}{*}{$O(m_1N_{b,1}^2)$}} & \multicolumn{1}{c}{\multirow {2}{*}{$1$}}\\
         \multicolumn{1}{c}{} & \multicolumn{1}{c}{$\mathbf y_c-\sum_{i=1}^{i=m_1}\widehat{\mathbf{A}}_{ci}\widehat{\mathbf{A}}_{i,i}^{-1}\mathbf y_i$ } & \multicolumn{1}{c}{} & \multicolumn{1}{c}{} \\
	      \midrule 
	        $4$ & $\left(\widehat{\mathbf{A}}_c-\sum_{i=1}^{i=m_1}\widehat{\mathbf{A}}_{ci}\widehat{\mathbf{A}}_{i,i}^{-1}\widehat{\mathbf{A}}_{ci}^H\right)^{-1}\left(\mathbf y_c-\sum_{i=1}^{i=m_1}\widehat{\mathbf{A}}_{ci}\widehat{\mathbf{A}}_{i,i}^{-1}\mathbf y_i\right)$ & $O(N_{b,1}^3)$ & $1$\\
	          \midrule 
	        $5$ & $\mathbf y_i-\widehat{\mathbf{A}}_{ci}^H\mathbf \omega_c$ & $O(\frac{L_1}{m_1}N_{b,1})$ & $m_1$\\
	          \midrule 
	        $6$ & $\widehat{\mathbf{A}}_{i,i}^{-1}\left(\mathbf y_i-\widehat{\mathbf{A}}_{ci}^H\mathbf \omega_c\right)$ & $O(N_{d,1}^2)$ & $m_1$\\
	       \bottomrule
	        \end{tabular}
	        \label{table:t3}
	\end{table*}
From equations (\ref{eqn:x1}) and (\ref{eqn:x2}), we draw the following conclusions. First, $\omega_c$, the sub-vector corresponding to the cut-node block, can be calculated independently of the other sub-vector $\omega_i$. Second, with $\omega_c$ obtained from (\ref{eqn:x1}), we can calculate each $\omega_i$ using (\ref{eqn:x2}) independently. The calculation of $\omega_i$ only involves the $i^{\text{th}}$ diagonal block of $\widehat{\mathbf{A}}$ and the corresponding $i^{\text{th}}$ border block $\widehat{\mathbf{A}}_{ci}$. In other words, if we treat each diagonal block as a cluster, then the signals received by each cluster can be processed in parallel of each other, while the interactions between different clusters are captured by $\omega_c$ and the border blocks.  	
  Based on the above discussions, Fig.~\ref{parallel} shows the architecture of the C-RAN BBU pool, where parallel signal processing is carried out. The arrows in Fig.~\ref{parallel} indicate the data flows between the processing units. As the figure shows, to expedite the calculation of $\omega_c$, matrices $\widehat{\mathbf{A}}_{ci}\widehat{\mathbf{A}}_{i,i}^{-1}\widehat{\mathbf{A}}_{ci}^H$ and vectors $\widehat{\mathbf{A}}_{ci}\widehat{\mathbf{A}}_{i,i}^{-1}\mathbf{y}_i$ are calculated at the same time by a number of parallel processing units and then fed into the central processing units. Then, $\omega_c$ is calculated in a central processing unit. The result is fed back into the parallel processing units. Each parallel processing unit is responsible for processing one cluster and calculating $\omega_i$. Specifically, we divide all the operations of signal processing in the proposed clustering algorithm into six steps as listed in Table~\ref{table:t3}. Steps 1 and 2 are first carried out in the parallel processing units. After receiving the results of steps 1 and 2, the central processing unit performs steps 3 and 4. At last, steps 5 and 6 are carried out in the parallel processing units.
  \vspace{-0.5cm}
\subsection{Optimizing the Computation Time}
Table~\ref{table:t3}\footnote{There are several matrix multiplication/inversion algorithms, which lead to various computational complexity. In this paper, we only take the complexity of some common algorithms as an example.} also lists the detailed computational complexity of each step in the single-layer DNC algorithm, where $N_{d,1}$ and $N_{b,1}$ are the average size of the diagonal blocks and cut-node block respectively. $m_1$ is the average number of diagonal blocks. $L_1 \ll N$ is the average number of non-zero entries per row in $\widehat{\mathbf{A}}$, which is an increasing function of the distance threshold $d_0$. In practical, $L_1$ does not increase with $N$, and is much smaller than $N$. Then, the total computational complexity in the parallel processing units is $O(N_{d,1}^3)$. The complexity in the central one is $O(N_{b,1}^{3})$. 
\par 
Before optimizing the computation time, we make some assumptions on the C-RAN BBU pool. We notice that the size of diagonal block should be determined by the processing power of the corresponding parallel processing unit. This implies that the sub-areas should have different side lengths, say different $r_1$. To simplify later discussions, we assume that all the parallel processing units in the BBU pool have equal processing power. We also notice that the central processing unit should be more powerful than the parallel ones. Otherwise, the central processing unit is unnecessary. For example, the corresponding operations, i.e., steps 3 and 4, can be carried out at one of the parallel processing units instead of the central one, and the total computational complexity can be reduced. Then, we define an unbalanced processing power ratio $\varrho$ to represent the processing power of the central processing unit and parallel ones. That is, to perform a same operation, the computation time of the central processing unit is $\varrho$ times shorter than that of the parallel ones. Denote a log-N ratio as $s=\log_N \varrho$ for notational brevity. Without loss of generality, the processing power of a parallel processing unit is normalized to be $1$. Then, the processing power of the central one is $N^s$. As the operations in steps 1, 2, 5 and 6 can be performed in parallel, the total computation time is $O(N_{d,1}^3+N_{b,1}^{3}N^{-s})$.
\par 
The computation time is an increasing function of the block sizes, $N_{d,1}$ and $N_{b,1}$. To achieve a short computation time, $N_{d,1}$ and $N_{b,1}$ should be as small as possible. However, the block sizes cannot be adjusted arbitrarily. In fact, for a given $r_1$, there is a fixed ratio between $N_{d,1}$ and $N_{b,1}$. We denote this ratio as $N^{z_1}=\frac{N_{d,1}}{N_{b,1}}$. Specifically, since $N_{d,1}$ and $N_{b,1}$ equal to the average number of RRHs in the sub-area center and the boundaries respectively, the relationship between $r_1$ and the ratio $N^{z_1}$ is 
\begin{equation}
\small
(r_1-2d_0)^2=4(r_1-d_0)d_0\frac{r^2}{r_1^2}N^{z_1},\label{eqn:r1vsz1}
\end{equation}
where $\beta_N$ is the RRH density. By adjusting $r_1$ from $2d_0$ to $r$, $z_1$ goes from $-1$ to $1$. Based on (\ref{eqn:r1vsz1}), we obtain the following approximations of $N_{d,1}$, $N_{b,1}$ and $m_1$:
\begin{lemma}\label{lemma:1}
In large C-RANs, given the block size ratio $N^{z_1}$, the approximations of $N_{d,1}$, $N_{b,1}$ and $m_1$ are 
\begin{equation}
\small
N_{d,1} \approx \left(4d_0\beta_N^{\frac{1}{2}}N^{1+z_1}\right)^{\frac{2}{3}}
\end{equation}
\begin{equation}
\small
N_{b,1} \approx \left(4d_0\beta_N^{\frac{1}{2}}N^{1-\frac{z_1}{2}}\right)^{\frac{2}{3}}
\end{equation}
\begin{equation}
\small
m_1  \approx (4d_0)^{-\frac{2}{3}}\beta_N^{-\frac{1}{3}}N^{\frac{1}{3}-\frac{2}{3}z_1}
\end{equation}
\end{lemma}
\begin{proof}
$r_1$ is the solution of (\ref{eqn:r1vsz1}), and we have
\begin{equation}
\small
(4d_0r^2N^z)^{\frac{1}{3}} \leq r_1\leq (4d_0r^2N^z)^{\frac{1}{3}}+2d_0,
\end{equation}
Then, when $d_0$ is much smaller than $r_1$,
\begin{equation}
\small
N_{d,1}=\beta_N(r_1-2d_0)^2 \approx \left(4d_0\beta_N^{\frac{1}{2}}N^{1+z_1}\right)^{\frac{2}{3}},
\end{equation}
\begin{equation}
\small
N_{b,1}=N_{d,1}N^{-z_1}\approx \left(4d_0\beta_N^{\frac{1}{2}}N^{1-\frac{z_1}{2}}\right)^{\frac{2}{3}},
\end{equation}
\begin{equation} 
\small
m_1 = \frac{r^2}{r_1^2} \approx (4d_0)^{-\frac{2}{3}}\beta_N^{-\frac{1}{3}}N^{\frac{1}{3}-\frac{2}{3}z_1}.
\vspace{-0.5cm}
\end{equation}
\end{proof}
After ignoring $d_0$, $\beta_N$ and $L_1$, we obtain the optimal computation time with parallel computing below.
\begin{lemma}\label{lemma:2}
When the log-N ratio $s \leq 3$, the minimum computation time with parallel computing is $O(N^{2-\frac{2}{3}s})$, with the optimal  $z_1=-\frac{s}{3}$. When $s>3$, the minimum computation time is $O(N^{3-s})$, with the optimal  $z_1=-1$.
\end{lemma}
\vspace{-0.3cm}
\begin{remark}\label{remark:2}
we notice that when the central processing unit is much more powerful than other units in the data center, performing all the operations in the central processor can achieve a shorter computation time than parallel computing. Based on Table~\ref{table:t3} and Lemma 1, the computation time of serial computing at the central processing unit is $O(N^{\frac{15}{7}-s})$ with $z_1=-\frac{1}{7}$.
\end{remark}
\par 
Then, based on Lemma \ref{lemma:2} and Remark \ref{remark:2}, we show the minimum computation time in Proposition 1.
\begin{proposition}
In C-RAN with parallel processing units and a central processing unit, when the log-N ratio $s\leq {\frac{3}{7}}$, the optimal computation time, $O(N^{2-\frac{2}{3}s})$, is achieved by parallel computing. When $s> {\frac{3}{7}}$, the optimal computation time, $O(N^{\frac{15}{7}-s})$, is achieved by performing all the operations at the central processing unit in serial. 
\end{proposition}
\vspace{-0.3cm}
\begin{figure}[!h]
	\centering
	{\includegraphics[width=0.46\textwidth]{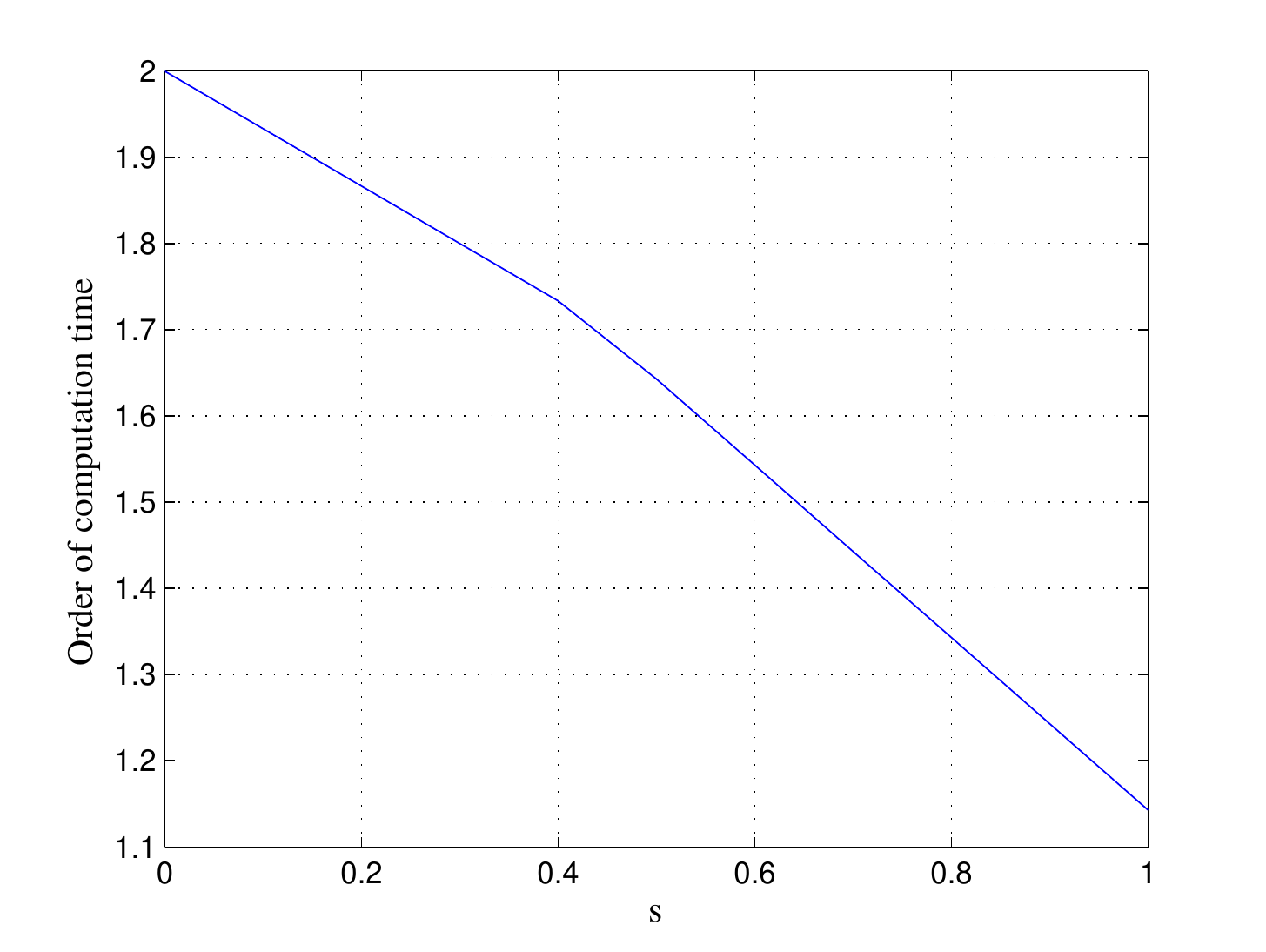}}
	\caption{Order of computation time vs log-N ratio $s$.}
	\vspace{-0.5cm}
	\label{supernode1}	
\end{figure}
\par 
We also show the effect of $s$ on the total computation time in Fig.~\ref{supernode1}, where the order of computation time is the maximum exponent of the computation time function. For example, the order of $O(N^{2-\frac{2}{3}s})$ is $2-\frac{2}{3}s$. We see that the order of the computation time is reduced with the increase of $s$. However, obviously, the price of the central processing unit is increased with the increase of $s$. Then, Fig.~\ref{supernode1} illustrates a trade-off between the computation time and the economic cost. This trade-off can serve as a guideline during the deployment of BBU pool. For example, when the economic cost is a major concern in the C-RAN system or a long computation time can be tolerated, processing units with low processing power, which leads to a low price, should be deployed. When the computation time is more important than cost, more powerful processing units should be selected.
\par
\begin{remark} 
So far, we have assumed that there are always enough parallel processing units, regardless of $r_1$ or $z_1$. In this case, we only need to optimize the sizes of diagonal blocks and the cut-node block and ignore the number of blocks. Instead, when the number of parallel processing units is limited, the number of blocks, $m_1$, also has an effect on the total computation time. Based on Lemma 1, $m_1$ can also be adjusted by $r_1$ or $z_1$. In this way, we can balance the computation time with limited availability of processing units. More detailed analysis, however, is out of the scope of this paper. 
\end{remark}
\par 
In this section, we have shown that the simple RRH labelling algorithm allows us to easily optimize the size of diagonal blocks , which can be directly interpreted as the size of clusters in large C-RANs. 
The result further provides an important guideline as to the architecture design of the C-RAN BBU pool, including the number of processing units, the choice between parallel and serial processing, the allocation of processing power among  BBUs, etc.
\vspace{-0.2cm}
\section{Multi-layer DNC Algorithm}
In the preceding section, we propose a single-layer DNC algorithm, which reduces the total computation time from $O(N^3)$ to $O(N^2)$. In this section, we propose a multi-layer DNC algorithm to further reduce the computation time. 
	\begin{figure}[!h]
	\centering
	{\includegraphics[width=0.5\textwidth]{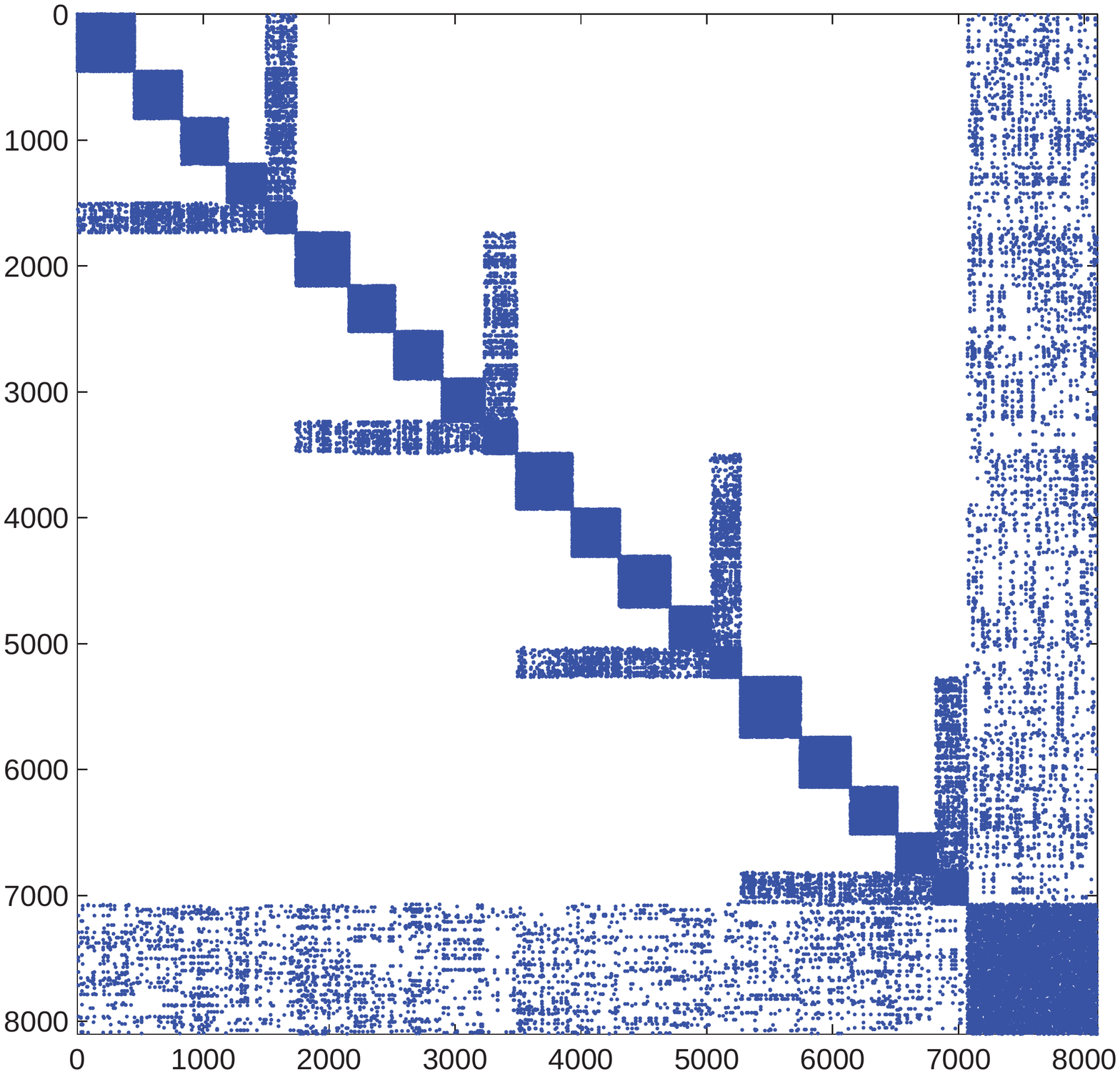}}
	\caption{$\widehat{\mathbf{A}}$ in a two-layer nested DBBD form after the second time RRH labelling, with $r_1=15$km, $r_2=8$km, where $N=8100, r=30$km, $d_0=500$meter.}\label{DBBD1}
\end{figure}
\begin{figure}[!h]
	\centering
	{\includegraphics[width=0.46\textwidth]{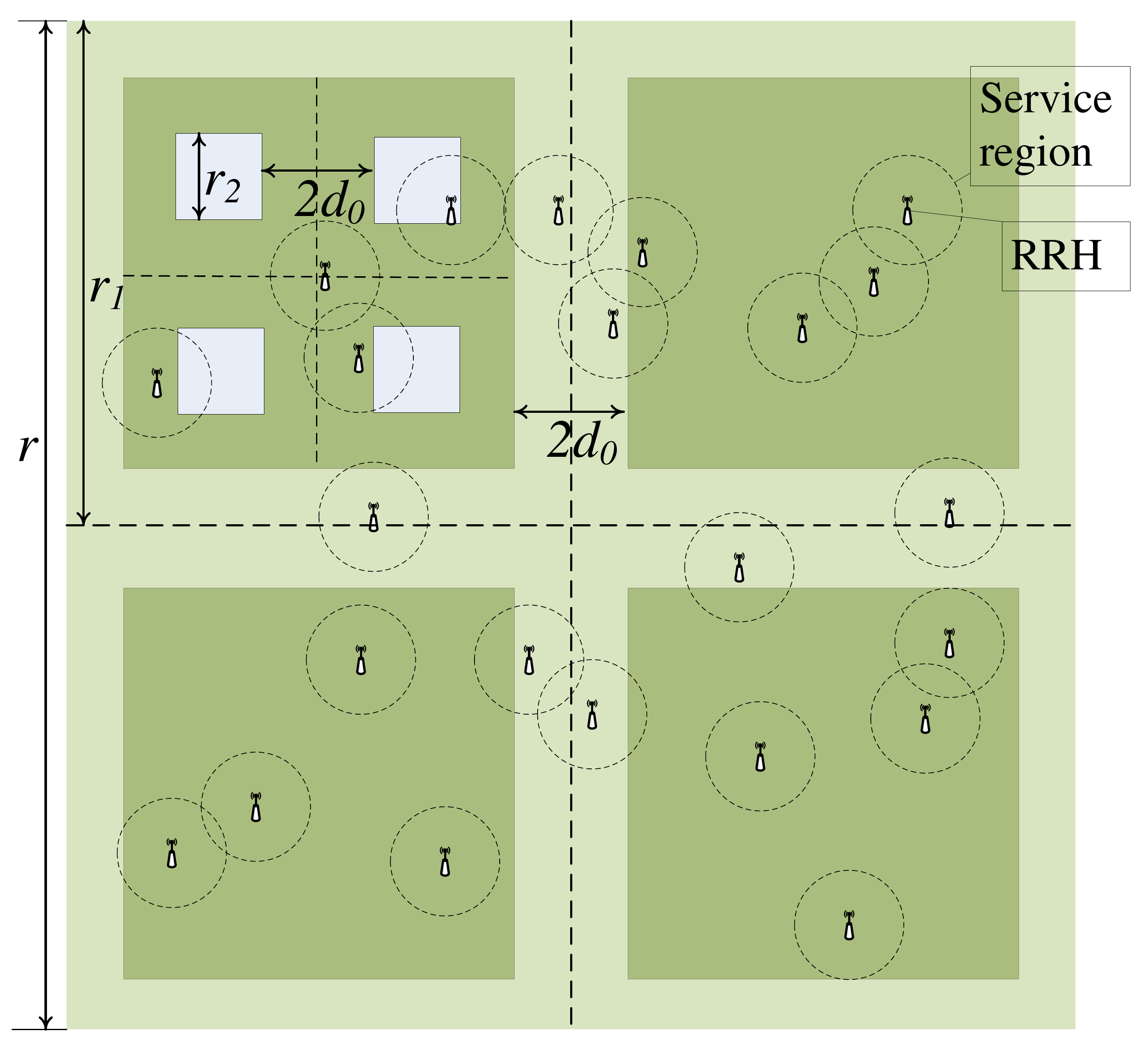}}
	\caption{Geographical RRH grouping in C-RAN.}\label{location2}
	\end{figure}
	\vspace{-0.3cm}
\par 
We notice that the computation time of the parallel processing units in the single-layer DNC algorithm is dominated by calculating $\widehat{\mathbf{A}}_{i,i}^{-1}$. Interestingly, a close study indicates that the diagonal blocks $\widehat{\mathbf{A}}_{i,i}$ are themselves sparse matrices. This is because the RRHs in the same cluster only have interactions with their neighboring RRHs instead of the whole cluster. This implies that $\widehat{\mathbf{A}}_{i,i}$ can be further permuted to a DBBD form, and the computation time of calculating $\widehat{\mathbf{A}}_{i,i}^{-1}$ can be reduced. As such, matrix $\widehat{\mathbf{A}}$ becomes a two-layer nested DBBD matrix. An example of such a matrix is shown in Fig.~\ref{DBBD1}.
\par
Fig.~\ref{location2} illustrates the RRH labelling strategy that turns $\widehat{\mathbf{A}}_{i,i}$ into a DBBD form. In particular, the RRHs in a cluster is grouped into sub-clusters. For example, for the top-left square,
the RRHs at the dark green boundary area are clustered to the sub-border, and the RRHs at the center area (the white area) are clustered to diagonal blocks. Intuitively, one can minimize the computation time by balancing the sizes of different blocks. This is the focus of our study in the remainder of this section.
\par 
Note that by repeating the process, $\widehat{\mathbf{A}}$ can be further permuted into a multi-layer nested DBBD matrix. For simplicity, we focus on the two-layer DNC algorithm in this section. The results, however, can be easily extended to the multi-layer case, as briefly discussed at the end of the section.  
\vspace{-0.5cm}
\subsection{Multi-Layer DNC Algorithm with Parallel Computing}
As discussed in the previous section, each parallel processor in Fig.~\ref{parallel} needs to calculate $\widehat{\mathbf{A}}_{i,i}^{-1}$. In the following, we show how $\widehat{\mathbf{A}}_{i,i}^{-1}$ can be computed in parallel, if $\widehat{\mathbf{A}}$ is already a two-layer nested DBBD matrix with diagonal blocks $\widehat{\mathbf{A}}_{i,i}$ being DBBD as well. For notational brevity, we denote the diagonal block $\widehat{\mathbf{A}}_{i,i} \in \mathcal C^{L \times L}$ by $\mathbf{B}$. Then, inverting $\mathbf{B}$ is equivalent to solving the following system:
\begin{equation}
\small
\begin{bmatrix}
     \mathbf{B}_{1,1} &  & & &\mathbf{B}_{c1}^H \\[0.3em]
       &\mathbf{B}_{2,2} & & &\mathbf{B}_{c2}^H\\[0.3em]
        & &\cdots & &\vdots\\[0.3em]
         & & &\mathbf{B}_{m,m} &\mathbf{B}_{cm}^H\\[0.3em]
       \mathbf{B}_{c1}& \mathbf{B}_{c2}&\cdots&\mathbf{B}_{cm} &\mathbf{B}_{c}
     \end{bmatrix}
     \begin{bmatrix}
     \mathbf X_{1}\\[0.3em]
     \mathbf X_{2}\\[0.3em]
     \vdots\\[0.3em]
     \mathbf X_{m}\\[0.3em]
     \mathbf X_{c}
     \end{bmatrix}
     =
     \begin{bmatrix}
     \mathbf I_{1}\\[0.3em]
     \mathbf I_{2}\\[0.3em]
     \vdots\\[0.3em]
     \mathbf I_{m}\\[0.3em]
     \mathbf I_{c}
     \end{bmatrix},
\end{equation}
where $\mathbf{X}=\left[\mathbf{X}_1^T, \mathbf{X}_2^T,\cdots,\mathbf{X}_m^T,\mathbf{X}_c^T\right]^T$ and $\mathbf{I}=\left[\mathbf{I}_1^T, \mathbf{I}_2^T,\cdots,\mathbf{I}_m^T,\mathbf{I}_c^T\right]^T$, with $\mathbf X_{i}$, $\mathbf I_{i} \in \mathcal{C}^{n_i \times L}$, and $\mathbf{X}_c$, $\mathbf{I}_c \in \mathcal{C}^{n_c \times L}$.
\par 
The columns in matrix $\mathbf{X}$ is given below:
\begin{equation}
\small
\begin{aligned}
\mathbf X_{c}
=&
\left(\mathbf{B}_c-\sum_{i=1}^{i=m}\mathbf{B}_{ci}\mathbf{B}_{i,i}^{-1}\mathbf{B}_{ci}^H\right)^{-1}\left[\mathbf{B}_{c1}\mathbf{B}_{1,1}^{-1},\cdots,\mathbf{B}_{c1}\mathbf{B}_{m,m}^{-1},\mathbf{I}\right],
\end{aligned}\label{eqn:inversea1}
\end{equation}
\begin{equation} 
\small
\begin{aligned}
&\mathbf X_{i}
=&\mathbf{B}_{i,i}^{-1}\left(\mathbf I_{i}-\mathbf{B}_{ci}^H\mathbf X_{c}\right).
\end{aligned}\label{eqn:inversea2}
\end{equation}
\begin{figure*}[!t]
	\centering
	{\includegraphics[width=0.95\textwidth]{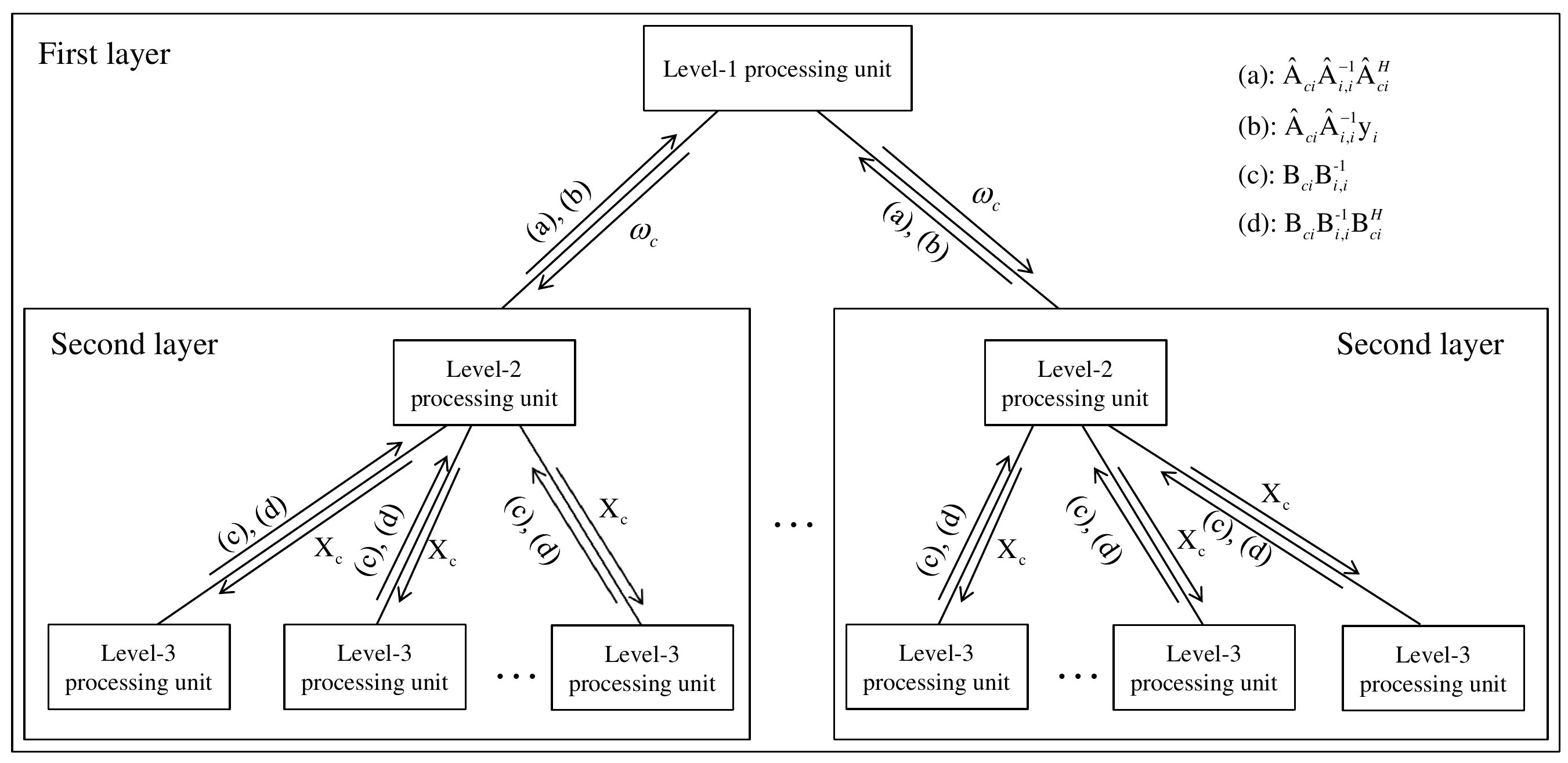}}
	\caption{Nested architecture of parallel computing of the two-layer DNC algorithm.}\label{parallel_2}
	\end{figure*}
\par 
\vspace{-0.7cm}
Similar to equations (\ref{eqn:x1}) and (\ref{eqn:x2}), parallel computing can be adopted in calculating (\ref{eqn:inversea1}) and (\ref{eqn:inversea2}). Combining the parallel computing of calculating $\widehat{\mathbf{A}}_{i,i}^{-1}$ in the second layer and that of solving $\omega$ in the first layer, a nested parallel computing architecture is illustrated in Fig.~\ref{parallel_2}. We first need to calculate $\widehat{\mathbf{A}}_{i,i}^{-1}$, i.e. $\mathbf{B}^{-1}$. Since the diagonal blocks in $\widehat{\mathbf{A}}$ are in DBBD forms, the calculation of $\mathbf{B}^{-1}$, can be split and allocated to a number of parallel processing units. As listed in Table~\ref{table:t4}, the calculation of $\mathbf{B}^{-1}$ (i.e. step 1 in Table~\ref{table:t3}) is divided into six steps. Steps 1.1 and 1.2 are first carried out in the level-3 processing units. The results are fed back into the level-2 processing units, which are responsible for performing steps 1.3 and 1.4. Then, steps 1.5 and 1.6 are carried out in the level-3 parallel processing units. 
Then, similar to the single-layer DNC algorithm, the level-2 processing units calculate matrices $\widehat{\mathbf{A}}_{ci}\widehat{\mathbf{A}}_{i,i}^{-1}\widehat{\mathbf{A}}_{ci}^H$ and vectors $\widehat{\mathbf{A}}_{ci}\widehat{\mathbf{A}}_{i,i}^{-1}\mathbf{y}_i$, and the results are fed into the level-1 processing unit. Then, $\omega_c$ is calculated by the level-1 processing unit, and $\omega_i$ is calculated by the level-2 processing unit. 
 \vspace{-0.5cm}
\subsection{Optimizing the Computation Time}
\par
\begin{table*}[!t]
\small
	\label{tab:parameter}
	        \caption{computation time of each step in equation (\ref{eqn:inversea1}) and (\ref{eqn:inversea2})}
	        \centering
	        \begin{tabular}{c  c  c  c }
	        \toprule
	   		step & operation & complexity/operation & total number of operations\\
	        \midrule 
	        $1.1$ & $\mathbf{B}_{i,i}^{-1}$ & $O(N_{d,2}^3)$ & $m_2$\\
	         \midrule 
	          \multicolumn{1}{c}{\multirow {2}{*}{$1.2$}} & \multicolumn{1}{c}{$\mathbf{B}_{ci}\mathbf{B}_{i,i}^{-1}\mathbf{B}_{ci}^H$} & \multicolumn{1}{c}{\multirow {2}{*}{$O(\frac{L_2}{m_2}N_{b,2}N_{d,2})$}} & \multicolumn{1}{c}{\multirow {2}{*}{$m_2$}}\\
         \multicolumn{1}{c}{} & \multicolumn{1}{c}{$\mathbf{B}_{ci}\mathbf{B}_{i,i}^{-1}\mathbf I_{i}$ } & \multicolumn{1}{c}{} & \multicolumn{1}{c}{} \\    
	        \midrule 
	          \multicolumn{1}{c}{\multirow {2}{*}{$1.3$}} & \multicolumn{1}{c}{$\mathbf{B}_c-\sum_{i=1}^{i=m_2}\mathbf{B}_{ci}\mathbf{B}_{i,i}^{-1}\mathbf{B}_{ci}^H$} & \multicolumn{1}{c}{\multirow {2}{*}{$O(m_2N_{b,2}^2)$}} & \multicolumn{1}{c}{\multirow {2}{*}{$1$}}\\
         \multicolumn{1}{c}{} & \multicolumn{1}{c}{$\mathbf I_{c}-\sum_{i=1}^{i=m_2}\mathbf{B}_{ci}\mathbf{B}_{i,i}^{-1}\mathbf I_{i}$ } & \multicolumn{1}{c}{} & \multicolumn{1}{c}{} \\
	      \midrule 
	        $1.4$ & $\left(\mathbf{B}_c-\sum_{i=1}^{i=m_2}\mathbf{B}_{ci}\mathbf{B}_{i,i}^{-1}\mathbf{B}_{ci}^H\right)^{-1}\left(\mathbf I_{c}-\sum_{i=1}^{i=m_2}\mathbf{B}_{ci}\mathbf{B}_{i,i}^{-1}\mathbf I_{i}\right)$ & $O(N_{b,2}^2N_{d,1})$ & $1$\\
	          \midrule 
	        $1.5$ & $\mathbf I_{i}-\mathbf{B}_{ci}^H\mathbf \mathbf X_{c}$ & $O(\frac{L_2}{m_2}N_{b,2}N_{d,1})$ & $m_2$\\
	          \midrule 
	        $1.6$ & $\mathbf{B}_{i,i}^{-1}\left(\mathbf I_{i}-\mathbf{B}_{ci}^H\mathbf \mathbf X_{c}\right)$ & $O(N_{d,2}^2N_{d,1})$ & $m_2$\\
	       \bottomrule
	        \end{tabular}
	        \label{table:t4}
	\end{table*}
	The computation time of each step in calculating $\mathbf{B}^{-1}$ is listed in Table~\ref{table:t4}, where $N_{d,t}$ and $N_{b,t}$ are the average diagonal block size and the cut-node block size in the $t^{\text{th}}$ layer, respectively, $m_2$ is the average number of diagonal blocks in $\mathbf{B}$, and $L_2 \ll N_{d,1}$ is the average number of non-zero entries per row in $\mathbf{B}$. Similarly to the single-layer DNC algorithm, we assume there exist enough processing units in each level, and the processing units in the same level have the same processing power. Moreover, the processing units in each level should be more powerful than those in higher levels. For example, processing units in Level 1 and 2 are more powerful than those in Level 3. Otherwise, we can remove Level 1 and 2 by shifting all the corresponding operations to Level 3. The reason is that the total number of processing units in Level 3 is always greater than that in Level 1 and 2. Then, we define an unbalanced processing power ratio $\varrho_t$ to represent the processing power ratio of processing units in the $t^{\text{th}}$ level and that in Level 3, where $t=1$ or $2$. Denote the corresponding log-N ratio as $s_t=\log_N \varrho_t$. Following the single-layer DNC algorithm, we can adjust the side length $r_t$ in the RRH labelling algorithm for different types of data centers as in the single-layer DNC algorithm. We first define the block-size ratio of a diagonal block and the cut-node block in the $t^{\text{th}}$layer as $N^{z_t}$, i.e., $N^{z_t}=\frac{N_{d,t}}{N_{b,t}}$. By adjusting $r_t$, $N^{z_t}$ goes from $N_{d,t-1}^{-1}$ to $N_{d,t-1}$. Then, the upper bounds of the diagonal block size and the cut-node block size in each layer is given below. 
\begin{lemma}\label{lemma:3}
In large C-RANs, given the block size ratio $N^{z_t}$ in the $t$th layer, the side length $r_t$ is the solution of equation $(r_t-2d_0)^2=4(r_t-d_0)d_0\frac{r^2}{r_t^2}N_{d,t-1}^{z_t}$. The approximations of $N_{d,t}$, $N_{b,t}$ and $m_t$ are 
\begin{equation}
\small
N_{d,t} \approx \left(4d_0\beta_N^{\frac{1}{2}}N_{d,t-1}N^{z_t}\right)^{\frac{2}{3}},
\end{equation}
\begin{equation}
\small
N_{b,t}\approx \left(4d_0\beta_N^{\frac{1}{2}}N_{d,t-1}N^{-\frac{z_t}{2}}\right)^{\frac{2}{3}},
\end{equation}
\begin{equation}
\small
m_t \approx \left((4d_0)^{-2}\beta_N^{-1}N_{d,t-1}N^{-2z_t}\right)^{-\frac{1}{3}}.
\end{equation}
\end{lemma}
\par 
Similar to the single-layer DNC algorithm, we find that when the processing units in Level 1 and Level 2 are much more powerful than those in Level 3, parallel computing is not the most efficient way. Here, we list three computing modes. 
\begin{itemize}
\item Mode 1: Steps 1.1, 1.2, 1.5, and 1.6 are executed at the level-3 processing units, steps 1.3, 1.4, 2, 5 and 6 are executed at the level-2 processing units, and steps 3 and 4 are carried out at the level-1 processing unit.
\item Mode 2: Steps 1, 2, 5 and 6 are executed at the level-2 processing units, and steps 3 and 4 are executed at the level-1 processing unit.
\item Mode 3: All the steps are executed at the level-1 processing unit in serial.
\end{itemize}
\par  
Then, we conclude the computing strategy of the two-layer DNC algorithm for different $s_1$ and $s_2$ in Proposition 2.
\begin{proposition}
For a three-level BBU pool with log-N ratio $s_1$ and $s_2$, 
\begin{itemize}
\item when $s_1+7s_2 < 3$ and $3s_1-2s_2<\frac{4}{3}$, choose mode 1, and the minimum computation time is $O(N^{\frac{42}{23}-\frac{14}{23}s_1-\frac{6}{23}s_2})$ by setting $z_1=\frac{4}{23}-\frac{9}{23}s_1+\frac{6}{23}s_2$ and $z_2=-\frac{1}{2}s_2$;
\item when $s_1+7s_2 \geq 3$ and $s_1-s_2<\frac{1}{3}$, choose mode 2, and the minimum computation time is $O(N^{\frac{15}{8}-\frac{5}{8}s_1-\frac{3}{8}s_2})$ with $z_1=\frac{1}{8}-\frac{3}{8}s_1+\frac{3}{8}s_2$ and $z_2=-\frac{3}{16}+\frac{1}{16}s_1-\frac{1}{16}s_2$;
\item when $3s_1-2s_2 \geq \frac{4}{3}$ and $s_1-s_2 \geq \frac{1}{3}$, choose mode 3, and the minimum computation time is $O(N^{2-s_1})$ with $z_1=0$ and $z_2=-\frac{1}{6}$. 
\end{itemize}
\end{proposition}
\vspace{-0.5cm}
\begin{figure}[!h]
	\centering
	{\includegraphics[width=0.46\textwidth]{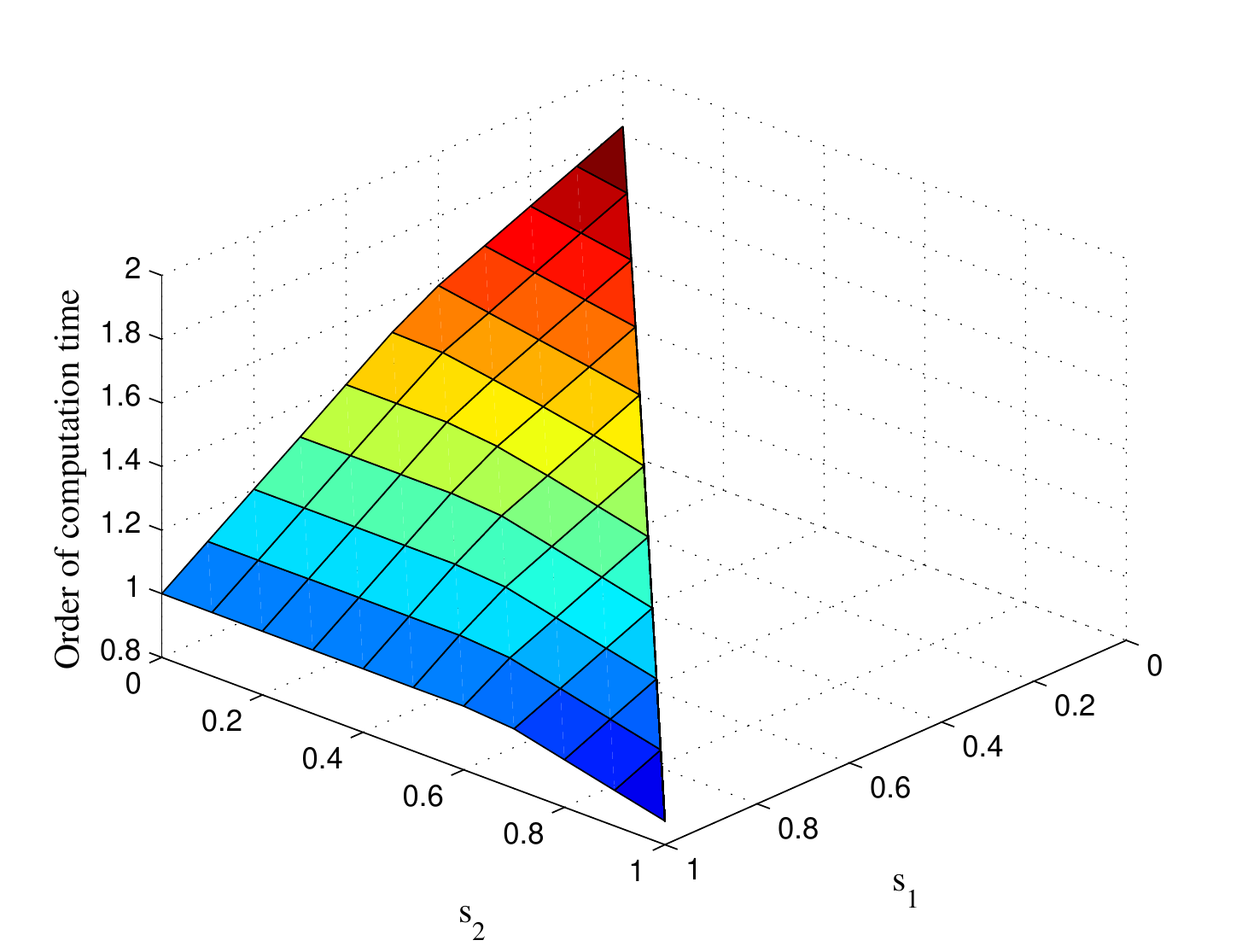}}
	\caption{Order of computation time vs log-N ratio $s_1$ and $s_2$.}\label{supernode2}	
\end{figure}
\par 
Repeating the parallel computing for the diagonal blocks, we can extend the DNC algorithm to multi-layer ones. We notice that the multi-layer DNC algorithm is more flexible for different data centers than the single-layer DNC algorithm. Multiple central processing units are allowed in the multi-layer DNC algorithm. Moreover, the total computation time is reduced in the multi-layer algorithm. Fig.~\ref{supernode2} shows the order of the computation time versus $s_1$ and $s_2$ for two-layer DNC algorithm. When $s_1=s_2=0$, the total computation time of the single-layer DNC algorithm is $O(N^2)$ while the computation time of the two-layer one is only $O(N^{\frac{42}{23}})$. It is not difficult to see that the computation time can be further reduced when more layers are introduced.
\section{Conclusions and Discussions}
In this paper, we proposed the DNC algorithm, based on which, both the channel estimation overhead and the computational complexity of the uplink signal processing can be significantly reduced. Moreover, our algorithm serves as a unified theoretical framework for dynamic clustering of C-RAN. 
By introducing a distance threshold, RRHs are grouped into disjoint center clusters or a boundary cluster based on their locations. The boundary cluster captures the interaction between different center clusters, which avoids the performance loss caused by conventional clustering. In addition, the operations in center clusters can be performed in parallel. We show that both the size and the number of the center clusters as well as the size of the boundary cluster can be easily adjusted. This further allows to flexibly balance the computation time, the number of parallel BBUs required, the allocation of computational power among BBUs, etc. Therefore, the DNC algorithm is adaptive to various architectures of the BBU pool.
\par 
It is obvious that the proposed DNC algorithm leads to a significant reduction in channel estimation overhead. That is because only the small scale fadings of the matrix entries that have not been discarded need to be estimated. However, in the proposed algorithm, we have not explicitly considered detailed channel estimation steps. In the future work, we will derive an effective algorithm to estimate the sparsified channel matrix with minimum channel estimation overhead and high estimation accuracy. In addition, we notice that, even though there is a small decrease in SINR caused by channel matrix sparsification, the effective system capacity may increase due to the reduction of channel estimation overhead. Based on this, we will design optimization algorithms to find the optimal distance threshold that maximizes the effective system capacity.

\section*{Appendix: proof of Theorem 1}
\vspace{-0.6cm}
\begin{subequations}	
\small
	\begin{align}
	&\mathrm{E}\left[\widehat{\text{SINR}}_k\left(d_0\right)\right]\\
	=&\mathrm{E}\left[\frac{P_k\left(|\widehat{\mathbf{v}}_k^H\widehat{\mathbf{h}}_k|^2+\widehat{\mathbf{v}}_k^H\left(\widehat{\mathbf{h}}_k\widetilde{\mathbf{h}}_k^H+\widetilde{\mathbf{h}}_k\widehat{\mathbf{h}}_k^H+\widetilde{\mathbf{h}}_k\widetilde{\mathbf{h}}_k^H\right)\widehat{\mathbf{v}}_k\right)}{\widehat{\mathbf{v}}_k^H\left(\sum_{j\neq k}P_j\mathbf{h}_j\mathbf{h}_j^H+N_0\mathbf{I}\right)\widehat{\mathbf{v}}_k}\right]\\
	\geq &\mathrm{E}_{\widehat{\mathbf{H}}}{\mathrm{E}_{\widetilde{\mathbf{h}}_j,\forall j \neq k}\left[\frac{P_k|\widehat{\mathbf{v}}_k^H\widehat{\mathbf{h}}_k|^2}{\widehat{\mathbf{v}}_k^H\left(\sum_{j\neq k}P_j\mathbf{h}_j\mathbf{h}_j^H+N_0\mathbf{I}\right)\widehat{\mathbf{v}}_k}\right]}\label{sinr_lb}\\
	\geq & \mathrm{E}_{\widehat{\mathbf{H}}}\left[\frac{P_k|\widehat{\mathbf{v}}_k^H\widehat{\mathbf{h}}_k|^2}{\widehat{\mathbf{v}}_k^H\left(\sum_{j\neq k}P_j\widehat{\mathbf{h}}_j\widehat{\mathbf{h}}_j^H+N_1\mathbf{I}+N_0\mathbf{I}\right)\widehat{\mathbf{v}}_k}\right]\label{sinr_lln}\\
	= &\mathrm{E}_{\mathbf{H}}\left[\frac{1}{1-\widehat{\mathbf{v}}_k^H\widehat{\mathbf{h}}_k}-1\right]\\
	=&\mathrm{E}_{\widehat{\mathbf{H}}}\left[P_k\text{tr}\left(\widehat{\mathbf{h}}_k\widehat{\mathbf{h}}_k^H\left(\sum_{j \neq k}P_j\widehat{\mathbf{h}}_j\widehat{\mathbf{h}}_j^H+ N_1\mathbf{I}+N_0\mathbf{I}\right)^{-1}\right)\right]\\
	=&\mathrm{E}_{\widehat{\mathbf{h}}_j,\forall j \neq k}\left[P_k\text{tr}\left(\widehat{\mu}\mathbf{I}\left(\sum_{j \neq k}P_j\widehat{\mathbf{h}}_j\widehat{\mathbf{h}}_j^H+ N_1\mathbf{I}+N_0\mathbf{I}\right)^{-1}\right)\right]\\
	=&P_k\widehat{\mu}\mathrm{E}\left[\text{tr}\left(\sum_{j \neq k}P_j\widehat{\mathbf{h}}_j\widehat{\mathbf{h}}_j^H+ N_1\mathbf{I}+N_0\mathbf{I}\right)^{-1}\right],
	\end{align}
	\end{subequations}
	where (\ref{sinr_lln}) holds due to the Jensen's inequality.
	\par 
	Likewise, we obtain 
	\begin{equation}
	\small 
	\begin{aligned}
	&\mathrm{E}\left[\text{SINR}_k\right]
	=&P_k{\mu}\mathrm{E}\left[\text{tr}\left(\sum_{j \neq k}P_j{\mathbf{h}}_j{\mathbf{h}}_j^H+N_0\mathbf{I}\right)^{-1}\right].
	\end{aligned}				
	\end{equation}
	Then, the SINR ratio becomes
	\begin{subequations}
	\small
	\begin{align}
	&\rho(d_0)\\
	\geq &\frac{P_k\widehat{\mu}\mathrm{E}_{\widehat{\mathbf{h}}_j,\forall j \neq k}\left[\text{tr}\left(\sum_{j \neq k}P_j\widehat{\mathbf{h}}_j\widehat{\mathbf{h}}_j^H+ N_1\mathbf{I}+N_0\mathbf{I}\right)^{-1}\right]}{P_k\mu\mathrm{E}_{\mathbf{h}_j,\forall j \neq k}\left[\text{tr}\left(\sum_{j \neq k}P_j\mathbf{h}_j\mathbf{h}_j^H+N_0\mathbf{I}\right)^{-1}\right]}\\
	\geq &\frac{P_k\widehat{\mu}\mathrm{E}_{\mathbf{h}_j,\forall j \neq k}\left[\text{tr}(\sum_{j \neq k}P_j\mathbf{h}_j{\mathbf{h}}_j^H+ N_1\mathbf{I}+N_0\mathbf{I})^{-1}\right]}{P_k{\mu}\mathrm{E}_{{\mathbf{h}}_j,\forall j \neq k}\left[\text{tr}\left(\sum_{j \neq k}{P_j\mathbf{h}}_j{\mathbf{h}}_j^H+N_0\mathbf{I}\right)^{-1}\right]}\label{MMSESINR_B} \\
	=& \frac{\hat \mu}{\mu} \mathrm{E}\left[\frac{\sum_{i=1}^{N}\frac{1}{\lambda_i+ N_1+ N_0}}{\sum_{i=1}^{N}\frac{1}{\lambda_i+ N_0}}\right]\label{MMSESINR_C} \\
	\geq & \frac{\hat \mu N_0}{\mu ( N_1 + N_0)}\label{MMSESINR_D},
	\end{align}
	\end{subequations}
where 
	$\lambda_1,\lambda_2,\cdots,\lambda_N$ are the eigenvalues of the positive semidefinite matrix $\sum_{j \neq k}P_j\mathbf{h}_j{\mathbf{h}}_j^H$. (\ref{MMSESINR_D}) holds since $ N_1 \geq 0$, $ N_0 \geq 0$ and $\lambda_i \geq0,  \forall i$.
	\par 
	Substituting $N_1=(\mu-\widehat{\mu})\sum_{j \neq k}P_k$ into (\ref{MMSESINR_D}), we have 
	\begin{equation}
	\small
	\rho(d_0)\geq \frac{\widehat{\mu}N_0}{\mu\left((\mu-\widehat{\mu})\sum_{j \neq k}P_k+N_0\right)}.
	\end{equation}
	When the users transmit equal power, i.e., $P_1=P_2=\cdots=P_K=P$, we have
	\begin{equation}
	\small
	\rho(d_0)\geq \frac{\widehat{\mu}N_0}{\mu\left((\mu-\widehat{\mu})(K-1)P+N_0\right)}.
	\end{equation}


						\end{document}